\documentclass[11pt]{scrartcl}
\usepackage{graphicx}

\usepackage[utf8]{inputenc} 
\usepackage{amsfonts, amssymb,amsmath,amsthm,bm}
\usepackage[british]{babel}
\usepackage{hyperref, tabularx, pbox, abstract, xpatch, csquotes}
\usepackage[outdir=./]{epstopdf}   
\usepackage[T1]{fontenc} 
\usepackage[caption=false]{subfig}
\usepackage[]{geometry}
\usepackage[format=plain]{caption}

\newcommand{\R}{\ensuremath{\mathbb{R}}}
\newcommand{\C}{\ensuremath{\mathbb{C}}}
\newcommand{\N}{\ensuremath{\mathbb{N}}}

\newcommand{\dif}{{\operatorname{d}}}

\newcommand{\tr}{{\operatorname{tr}}}

\newtheorem{prop}{Proposition}

\newtheorem{cor}[prop]{Corollary}

\newcommand{\mattwo}[4]{\left(
	\begin{array}{cc}
		#1 & #2 \\
		#3 & #4 \\  
	\end{array}
	\right)
}

\newcommand{\vectwo}[2]{\left(
	\begin{array}{c}
		#1\\
		#2\\  
	\end{array}
	\right)
}
\numberwithin{equation}{section}
\numberwithin{figure}{section}
\numberwithin{table}{section}

\usepackage{xcolor}
\definecolor{changes}{RGB}{0, 0, 0}

\begin{document}

\title{An integrodifference model for vegetation patterns in semi-arid environments with seasonality}


\author{L. Eigentler and J. A. Sherratt}

\newcommand{\Addresses}{{
		\bigskip
		\footnotesize
		
		L. Eigentler (corresponding author),  Maxwell Institute for Mathematical Sciences, Department of Mathematics, Heriot-Watt University, Edinburgh EH14 4AS, United Kingdom\par\nopagebreak
		\textit{E-mail address}: \texttt{le8@hw.ac.uk}
		
		\medskip
		
		J.A. Sherratt,  Maxwell Institute for Mathematical Sciences, Department of Mathematics, Heriot-Watt University, Edinburgh EH14 4AS, United Kingdom\par\nopagebreak
		\textit{E-mail address}: \texttt{J.A.Sherratt@hw.ac.uk}}}
\date{}

\maketitle
\Addresses

\begin{abstract}
	Vegetation patterns are a characteristic feature of semi-deserts occurring on all continents except Antarctica. In some semi-arid regions, the climate is characterised by seasonality, which yields a synchronisation of seed dispersal with the dry season or the beginning of the wet season. We reformulate the Klausmeier model, a reaction-advection-diffusion system that describes the plant-water dynamics in semi-arid environments, as an integrodifference model to account for the temporal separation of plant growth processes during the wet season and seed dispersal processes during the dry season. The model further accounts for nonlocal processes involved in the dispersal of seeds.
	Our analysis focusses on the onset of spatial patterns. The Klausmeier partial differential equations (PDE) model is {\color{changes}is linked to the integrodifference model in an appropriate limit}, which yields a control parameter for the temporal separation of seed dispersal events. We find that the conditions for pattern onset in the integrodifference model are equivalent to those for the continuous PDE model and hence independent of the time between seed dispersal events. We thus conclude that in the context of seed dispersal, a PDE model provides a sufficiently accurate description, even if the environment is seasonal. This emphasises the validity of results that have previously been obtained for the PDE model. 
	Further, we numerically investigate the effects of changes to seed dispersal behaviour on the onset of patterns. We find that long-range seed dispersal inhibits the formation of spatial patterns and that the seed dispersal kernel's decay at infinity is a significant regulator of patterning.
\end{abstract}


\section{Introduction}

Vegetation patterns are a ubiquitous feature of ecosystems in semi-arid climate zones. Occurrences of such mosaics of plants and bare soil have been reported from all continents except Antarctica, including the African Sahel \cite{Deblauwe2012} and the Horn of Africa \cite{Gowda2018}, Western Australia \cite{Gandhi2018}, northern Chile \cite{Fernandez-Oto2019a}, Israel \cite{Sheffer2013}, the Chihuahuan Desert in North America \cite{Deblauwe2012} and Southeastern Spain \cite{Lesschen2008}. A detailed understanding of the evolution of vegetation patterns is of considerable importance as they hold valuable information on the health of ecosystems. For example, changes to a pattern's properties such as its wavelength, its recovery time from perturbations, or the area fraction covered by biomass can act as  early warning signals of desertification \cite{Corrado2014, Dakos2011, Gowda2016, Kefi2007, Rietkerk2004, Saco2018, Zelnik2018}. Desertification processes are a major threat to economies in semi-deserts as agriculture provides a significant contribution to GDP \cite{UNGlobalLandOutlook2017}. For example, the livestock sector, which depends in part on animals grazing on spatially patterned vegetation, accounts for 20\% of GDP in Chad and involves 40\% of its population \cite{Dickovick2014, UNLivestockBrief2005}.

A number of feedback mechanisms may be involved in the pattern formation process (see \cite{Meron2018} for a review), but it is widely agreed that a central mechanism is the vegetation-infiltration feedback loop, which results in a redistribution of water towards areas of high biomass. On bare soil, the formation of physical and biological soil crusts inhibits water infiltration into the soil \cite{Eldridge2000}. Thus, water run-off towards existing vegetation patches occurs. The enhancement of environmental conditions in these sinks for the limiting resource drives further plant growth and thus closes the feedback loop \cite{Thompson2010}.

Dryland plants have developed a range of seed production and dispersal strategies to cope with the environmental stress in their habitats \cite{Ellner1981, RheedevanOudtshoorn2013}. One such mechanism, commonly observed in water-controlled ecosystems, is ombrohydrochory, the dispersal of seeds caused by an opening of the seed container due to contact with water \cite{Parolin2006, Navarro2009, RheedevanOudtshoorn2013}. One particular form, exhibited by members of the Aizoaceae family in semi-arid regions of the Sahel, Australia and South America, is ballistic dispersal, which uses the kinetic energy of raindrops to expulse the plants' seeds \cite{Parolin2006, Friedman1978}. Some semi-arid environments such as those in the Mediterranean are characterised by seasonal fluctuations in their environmental conditions and in particular in their precipitation patterns \cite{Noy-Meir1973}. In combination with processes that allow plants to store diaspores during periods of drought, ombrohydrochory yields a synchronisation of seed dispersal with the beginning of the wet season in such seasonal environments. This synchronisation has, for example, been reported in \textit{Mesembryanthemum crystallinum} and \textit{Mesembryanthemum nodiflorum} in Southeastern Spain \cite{Navarro2009}. If seed dispersal strategies different from ombrohydrochory are dominant, most species disperse their seeds during the dry season \cite{Navarro2009,Shabana2018}.


The seasonal synchronisation of seed dispersal splits the annual life-cycle of a plant population into two distinct stages. During the wet season, seeds germinate, new seedlings emerge and adult plants increase their biomass, but no spatial movement takes place. Seed dispersal only occurs during, or at the end of the dry season, while growth processes are dormant \cite{Baudena2008}. {\color{changes}By} contrast, most mathematical models for dryland vegetation patterns consist of partial differential equations and thus assume that seed dispersal occurs continuously in time. A widely used approach to account for the temporal structure of the annual life cycle is the use of integrodifference equations. This splits the system into 2 distinct, non-overlapping phases, which are both described as discrete, instantaneous processes: a growth phase during which dispersal processes are either not present or negligible and a dispersal phase during which no growth occurs. 
The application of integrodifference equations to biological and ecological systems in which spatial dispersal plays a significant role was in part pioneered by Kot and Schaffer \cite{Kot1986}, and has become a well-established tool in the description of biological and ecological systems since then (e.g. \cite{Musgrave2014I, Musgrave2014II, Powell2004, Clark2003, Neubert1995}).

The spatial and temporal scales associated with the evolution of vegetation patterns do not allow their recreation in laboratory settings. Instead, a range of mathematical models have been proposed to address different aspects of the pattern dynamics \cite{Borgogno2009, Zelnik2013}. A significant amount of modelling work is based on systems of partial differential equations, most notably by Gilad et al. \cite{Gilad2004}, HilleRisLambers, Rietkerk et al. \cite{HilleRisLambers2001, Rietkerk2002} and Klausmeier \cite{Klausmeier1999}. The reaction-advection-diffusion Klausmeier model \cite{Klausmeier1999} is a deliberately basic description of dryland ecosystems based on the vegetation-infiltration feedback loop. Its relative simplicity provides a rich framework for model analyses and extensions (e.g. \cite{Bennett2019, Consolo2019, Eigentler2018nonlocalKlausmeier, Eigentler2019Multispecies, Siteur2014, Ursino2006, Sherratt2010, Sherratt2011, Sherratt2013III, Sherratt2013IV, Sherratt2013V}). The recent development of new remote sensing technology, using temporal sequences of satellite images, allows for comparisons between model predictions and field data \cite{Bastiaansen2018, Gandhi2018}.

In the Klausmeier model, seed dispersal is modelled by a diffusion term. In reality, the dispersal of seeds is affected by nonlocal processes, such as ballistic dispersal or long range dispersal (e.g. via mammals or wind) \cite{Pueyo2008, Bullock2017}. The Klausmeier model has been extended to account for such nonlocal processes \cite{Eigentler2018nonlocalKlausmeier, Bennett2019} and a similar approach has been applied to other models for dryland vegetation \cite{Baudena2013, Pueyo2008, Pueyo2010}. Integrodifference systems also provide a description of nonlocal dispersal effects through a convolution of the plant density with a kernel function. The kernel function is a probability density function describing the average distribution of seeds dispersed from a single plant. The dispersal kernel's properties (in particular its shape and standard deviation) depend on both plant species and environmental conditions \cite{Bullock2017}.

In this paper we address the significance of seed dispersal synchronisation and its temporal separation from growth processes in seasonal dryland environments. To do so, we introduce an integrodifference model describing the plant-water dynamics in semi-arid ecosystems in Section \ref{sec: models}. We base our model on the Klausmeier model, to compare our results to previous model analyses of models with no temporal structure. To aid comparisons to the PDE model, we review the most relevant results for the Klausmeier model in Section \ref{sec: models}. Even though an integrodifference model cannot explicitly take into account the length of the plant growth stage, a convergence result (Proposition \ref{prop: Difference: limit case}) yields a control parameter for the temporal separation of seed dispersal events through an appropriate parameter setting. In Section \ref{sec: integrodifference: linstab} we focus on this special case and perform a linear stability analysis to determine a condition for pattern onset in the model and investigate this condition under variations in the growth season length. The analytical derivation of this condition relies on a specific (but nevertheless biologically relevant) choice of the dispersal kernels. To relax this assumption we perform numerical simulations in Section \ref{sec:difference: simulations} to determine the parameter region in which pattern onset occurs for other biologically relevant dispersal kernels. Finally, we discuss our results in Section \ref{sec: difference: Discussion}.

\section{The Models}\label{sec: models} 
In this section we introduce the integrodifference model which we use to investigate the effects of seasonal synchronisation of seed dispersal on the onset of vegetation patterns in semi-arid environments. The model is based on the reaction-advection-diffusion model by Klausmeier \cite{Klausmeier1999} and to facilitate the comparison of our results on the discrete model to that of the time-continuous model, we start by reviewing relevant results for the Klausmeier model. We relate the models through a convergence result that shows that {\color{changes}solutions of the integrodifference model converge to solutions of the Klausmeier PDE model in an appropriate limit.}
\subsection{Klausmeier Model}
One of the well-established models describing vegetation patterns in semi-arid environments is the Klausmeier model \cite{Klausmeier1999}. It reduces the plant-water dynamics to a small set of basic processes (rainfall, plant mortality, evaporation/drainage, vegetation-infiltration feedback and spatial dispersal). The relative simplicity of this modelling approach provides a framework for a rich mathematical analysis (e.g. \cite{Sherratt2005,Sherratt2007,Sherratt2010,Sherratt2011,Sherratt2013III,Sherratt2013IV,Sherratt2013V,Siteur2014,Ursino2006}). Suitably nondimensionalised \cite{Klausmeier1999,Sherratt2005}, the model is
\begin{subequations}
	\label{eq: Intro Klausmeier local}
	\begin{align}
	\frac{\partial u}{\partial t} &= \overbrace{u^2w}^{\text{plant growth}} - \overbrace{Bu}^{\text{plant mortality}} + \overbrace{\frac{\partial^2 u}{\partial x^2}}^{\text{plant dispersal}}, \label{eq:Intro Klausmeier local plants}   \\
	\frac{\partial w}{\partial t} &= \underbrace{A}_{\text{rainfall}}-\underbrace{w}_{\substack{\text{evaporation}\\\text{and drainage}}}-\underbrace{u^2w}_{\substack{\text{water consumption}\\\text{by plants}}} + \underbrace{\nu\frac{\partial w}{\partial x}}_{\substack{\text{water flow}\\\text{downhill}}}+ \underbrace{d\frac{\partial^2w}{\partial x^2}}_{\text{water diffusion}}.
	\end{align}
\end{subequations}
Here $u(x,t)$ denotes the plant density, $w(x,t)$ the water density, $x\in \R$ the space domain where $x$ is increasing in the uphill direction and $t>0$ the time. Originally, the model only focussed on a sloped spatial domain, but the addition of a water diffusion term to account for the possibility of a description on flat terrain is a well established addition \cite{Kealy2012,Siteur2014, Stelt2013, Zelnik2013}. To emphasise on the description of seed dispersal as a local process, we refer to this model as the ``local Klausmeier model'' throughout the paper. Water input to the system is assumed to occur at a constant rate, evaporation and drainage effects are proportional to the water density \cite{Rodriguez-Iturbe1999, Salvucci2001} and the plant mortality rate is density-independent. The nonlinearity in the description of water uptake and plant growth processes arises due to a soil modification by plants. The term is the product of the density of the consumer $u$ and of the available resource $uw$, the amount of water that is able to infiltrate into soil layers where plant roots consume water. The dependence on the plant density $u$ in the latter term occurs due to a positive correlation between the plant density and the soil surface's permeability \cite{Rietkerk2000, Valentin1999, Cornet1988}. Finally, plant growth is assumed to be proportional to the amount of consumed water \cite{Rodriguez-Iturbe1999, Salvucci2001}. The parameters $A$, $B$, $\nu$ and $d$ are combinations of different dimensional parameters but can be interpreted as rainfall, plant loss, the slope and water diffusion, respectively. 

In a previous paper \cite{Eigentler2018nonlocalKlausmeier} we have introduced nonlocal seed dispersal effects to the model by replacing the plant diffusion term by a convolution of a dispersal kernel (a probability density function) $\phi$ and the plant density $u$. The resulting model is referred to as the ``nonlocal Klausmeier model'' and is
\begin{subequations}
	\label{eq: Intro Klausmeier nonlocal}
	\begin{align}
	\frac{\partial u}{\partial t} &= u^2w - Bu + C\left( \phi(\cdot) \ast u(\cdot,t) - u(x,t)\right), \label{eq:Intro Klausmeier nonlocal plants}   \\
	\frac{\partial w}{\partial t} &= A-w-u^2w + \nu\frac{\partial w}{\partial x}+ d\frac{\partial^2w}{\partial x^2}.
	\end{align}
\end{subequations}
The additional parameters $C$ and $a$ represent the rate of plant dispersal and reciprocal width of the dispersal kernel, respectively. {\color{changes} Note that the convolution $(\phi \ast u)(x,t)$ accounts for all plant biomass dispersed to the space point $x$, including the fraction of biomass that is not dispersed. The final term in \eqref{eq:Intro Klausmeier nonlocal plants} ensures that the total biomass over the whole domain remains unchanged by the seed dispersal term.} The nonlocal model \eqref{eq: Intro Klausmeier nonlocal} and the local model \eqref{eq: Intro Klausmeier local} are related through a convergence result. If the dispersal kernel $\phi$ is decaying exponentially as $|x|\rightarrow \infty$, then the local model \eqref{eq: Intro Klausmeier local} can be obtained from the nonlocal model \eqref{eq: Intro Klausmeier nonlocal} in the limit $C\rightarrow \infty$ and $\sigma \rightarrow 0$ with $C=2/\sigma^2$, where $\sigma$ denotes the standard deviation of $\phi$ \cite{Eigentler2018nonlocalKlausmeier}.

Linear stability analysis of both the local and the nonlocal Klausmeier model with the Laplace kernel
\begin{align}\label{eq:difference equations: simulations: Laplacian}
\phi(x) = \frac{a}{2}e^{-a|x|}, \quad a>0, x\in\R.
\end{align}
provides analytically derived conditions for pattern onset to occur in the system. On flat ground, i.e. $\nu = 0$, Turing-type patterns form due to a diffusion-driven instability, i.e. there exists a threshold $d_c>0$ on the diffusion coefficient such that an instability occurs for all $d>d_c$. In the local model \eqref{eq: Intro Klausmeier local}, the threshold is
\begin{align}\label{eq: Models: Klausmeier: flat ground diffusion threshold}
d_c(A,B) = \frac{8B\sqrt{-A^2+A\sqrt{A^2-4B^2}+4B^2}-2A^2+2A\sqrt{A^2-4B^2}+16B^2}{B\left(A-\sqrt{A^2-4B^2}\right)^2}.
\end{align}
{\color{changes} A corresponding threshold $\widetilde{d_c}(A,B,C,a)$ for the nonlocal model \eqref{eq: Intro Klausmeier nonlocal} with the Laplace kernel \eqref{eq:difference equations: simulations: Laplacian} can be derived explicitly, but it omitted due to its algebraic complexity.}

On sloped ground ($\nu \ne 0$) pattern onset has been studied close to a Turing-Hopf bifurcation, which is characterised by an upper bound on the rainfall parameter $A$ that has been derived analytically valid to leading order in $\nu$ as $\nu \rightarrow \infty$ for both models \cite{Eigentler2018nonlocalKlausmeier, Sherratt2013IV}. The calculation of this upper bound on the precipitation parameter for the nonlocal model with the Laplace kernel shows that long range dispersal of seeds inhibits the formation of patterns by decreasing the size of the parameter region that supports the onset of patterns. On flat ground an increase of the dispersal kernel's standard deviation causes an increase in the threshold on the diffusion coefficient, while on sloped ground an increase in the dispersal kernel's width inhibits the formation of patterns by decreasing the upper bound on the rainfall parameter. 

The analytical derivation of pattern onset conditions in the nonlocal model is facilitated by the simple algebraic form of the Laplace kernel's Fourier transform and the associated polynomial structure of the dispersion relation in the linear stability analysis. For other biologically relevant seed dispersal kernels, conditions for pattern onset are not analytically tractable. Numerical simulations, however, confirm the qualitative trends obtained for the model with the Laplace kernel. Simulations further suggest that the dispersal kernel's decay at infinity has an influence on the rainfall threshold. For narrow dispersal kernels, those that account for more rare long-range dispersal events (algebraic decay rather than exponential) have an inhibitory effect on the formation of patterns, while for sufficiently wide kernels those that decay algebraically at infinity promote pattern formation compared to exponentially decaying kernels.

\subsection{Integrodifference Model}\label{sec: Models: integrodifference} 
Integrodifference models are a common type of model widely used in the description of systems in which dispersal processes are temporally separated from other dynamics such as growth/birth and decay/death. To account for the separation of plant growth and seed dispersal stages in dryland ecosystems, we propose the integrodifference model
\begin{subequations}
	\label{eq: Klausmeier integrodifference}
	\begin{align}
	u_{n+1}(x) &= C \phi \ast f(u_n,w_n), \label{eq: Klausmeier integrodifference plants}\\
	w_{n+1}(x) &= D  \phi_1 \ast g(u_n,w_n), \label{eq: Klausmeier integrodifference water}
	\end{align}
\end{subequations}
where
\begin{align*}
f\left(u,w\right) &= u^2w-Bu + \frac{1}{C}u, \\ 
g(u,w) &= A - u^2w- w + \frac{1}{D}w.
\end{align*}
Here $u_n(x)$ denotes the plant density, $w_n(x)$ the water density after $2n, n\in\N$ seasons and location $x\in\R$, where $x$ increases in the uphill direction. The formulation of the model splits the processes involved into two phases: a growth and evolution phase described by the functions $f(u,w)$ and $g(u,w)$ during which no dispersal occurs, and a dispersal phase modelled as a convolution of the evolved densities with dispersal kernels. As in the nonlocal Klausmeier model \eqref{eq: Intro Klausmeier nonlocal}, the plant dispersal kernel $\phi$ is symmetric and represents isotropic dispersal of plants. To model the flow of water downhill, the water dispersal kernel $\phi_1$ is {\color{changes}in general asymmetric with mean $\mu_{\phi_1}\le 0$. The special case of a symmetric kernel $\phi_1$} corresponds to the model on flat ground, which is the main aspect of the study in this paper. The model is based on the Klausmeier models \eqref{eq: Intro Klausmeier nonlocal} and \eqref{eq: Intro Klausmeier local} and thus the functions $f(u,w)$ and $g(u,w)$ consist of the terms describing the rate of change in the original model, appropriately scaled by the coefficients $C$ and $D$ to reflect the time between steps in the discrete model, added to the existing densities.

{\color{changes}As the integrodifference model \eqref{eq: Klausmeier integrodifference} arises directly from the local Klausmeier model \eqref{eq: Intro Klausmeier local}, the two models can be linked through a consistency result in an appropriate limit which shows that \textit{the integrodifference model \eqref{eq: Klausmeier integrodifference} tends to the local Klausmeier model \eqref{eq: Intro Klausmeier local}} as $T\rightarrow 0$. To show this, we consider the parameter setting
	\begin{align}\label{eq:difference: limit scalings}
	C=T, \quad \sigma_\phi^2=2T, \quad D=T, \quad \mu_{\phi_1} = -\nu T, \quad \widetilde{\sigma}_{\phi_1}^2 = 2dT,
	\end{align}	
	where $\mu$ and $\sigma$ denote the mean and standard deviation of the respective kernels and $	\widetilde{\sigma}_{\phi_1}^2 =  \int_{-\infty}^\infty \phi_1(x)x^2 \dif x$, the second raw moment of the kernel function $\phi_1$. Further, we define operators $P, P_T: C^\infty(\R \times [0,\infty), [0,\infty)^2) \rightarrow C^\infty(\R \times [0,\infty), [0,\infty)^2)$ by
	\begin{align}\label{eq: Difference: consistency pde operator}
	P\bm{v}(x,t) = \frac{\partial \bm{v}}{\partial t} (x,t) - \Gamma\bm{v}(x,t) - h_1(\bm{v}(x,t)),
	\end{align}
	for any function $\bm{v}(x,t) = (u(x,t),w(x,t)) \in C^\infty(\R \times [0,\infty), [0,\infty)^2)$, where 
	\begin{align*}
	\Gamma = \operatorname{diag}\left(\frac{\partial^2}{\partial x^2}, \nu \frac{\partial}{\partial x} + d\frac{\partial^2}{\partial x^2} \right), \quad h_1\left(\bm{v}\right) = \vectwo{u^2w-Bu}{A-u^2w-w},
	\end{align*}
	and
	\begin{align}\label{eq: Difference: consistency difference operator}
	P_T\bm{v}(x,t) = \frac{1}{T} \left( \bm{v}(x,t+T) - h_2(\bm{v}(x,t))\right),
	\end{align}
	where 
	\begin{align*}
	h_2\left(\bm{v}(x,t)\right) = \vectwo{-C \phi(\cdot) \ast f(u(\cdot,t),w(\cdot,t))}{- D  \phi_1(\cdot) \ast g(u
		(\cdot,t),w(\cdot,t))}.
	\end{align*}
	Note that the operator $P$ arises from the local Klausmeier model \eqref{eq: Intro Klausmeier local}, because $P\bm{v} = 0$ for any $\bm{v}$ that satisfies \eqref{eq: Intro Klausmeier local}. Similarly, $P_T$ represents the integrodifference model \eqref{eq: Klausmeier integrodifference}, because a sequence $\bm{v}_n(x) = \bm{v}(x,nT)$ satisfies \eqref{eq: Klausmeier integrodifference} if $P_T\bm{v}_n = 0$ for all $n\in\N$.   Utilising this reformulation of both models, it is possible to show the following result.
}

\begin{prop}\label{prop: Difference: limit case}
	{\color{changes}Consider the parameter setting \eqref{eq:difference: limit scalings} and let the kernel functions $\phi$ and $\phi_1$ have finite moments of all orders and decay exponentially as $|x| \rightarrow \infty$.  Then the integrodifference model \eqref{eq: Klausmeier integrodifference} is consistent with the local Klausmeier model \eqref{eq: Intro Klausmeier local}, i.e. 	
		\begin{align*}
		P\bm{v} - P_T\bm{v} \rightarrow 0 \quad \text{as} \quad T\rightarrow 0^+,
		\end{align*}
		for any $\bm{v} \in C^\infty(\R \times [0,\infty), [0,\infty)^2)$.}
\end{prop}

{\color{changes}In other words, the model equations \eqref{eq: Klausmeier integrodifference} converge to the model equations \eqref{eq: Intro Klausmeier local} as $T\rightarrow 0^+$. The notion of \textit{consistency} is widely used in the field of numerical analysis, and crucially it does not imply convergence of model solutions. While we are unable to construct an argument to prove convergence, numerical simulations suggest that solutions of the integrodifference model \eqref{eq: Klausmeier integrodifference} converge to solutions of the local Klausmeier model \eqref{eq: Intro Klausmeier local} in the parameter setting \eqref{eq:difference: limit scalings} as $T\rightarrow 0^+$ (Fig. \ref{fig:difference equations: simulations: convergence simulation}).
}

On sloped ground Prop. \ref{prop: Difference: limit case} requires that $\nu = o(T^{-1})$, so that $T \nu \rightarrow 0$ as $T\rightarrow 0^+$ and $\nu \rightarrow \infty$, to facilitate any asymptotic analysis in $\nu$ similar to that of the local Klausmeier model \cite{Sherratt2005,Sherratt2010,Sherratt2011,Sherratt2013III,Sherratt2013IV,Sherratt2013V}. On flat ground, $\phi_1$ is symmetric and thus $\mu_{\phi_1}=0$ and $\widetilde{\sigma}_{\phi_1}$ coincides with the kernel's standard deviation $\sigma_{\phi_1}$. 

The parameter $T$ can be interpreted as the time between separate dispersal events and the scalings \eqref{eq:difference: limit scalings} are thus the main focus of the model's analysis in Section \ref{sec: integrodifference: linstab}. While the time between two seed dispersal events in a seasonal environment is usually fixed, we are interested in variations of $T$ as this parameter establishes a connection between the local Klausmeier model \eqref{eq: Intro Klausmeier local} and the integrodifference model \eqref{eq: Klausmeier integrodifference}. {\color{changes}In particular, as $T\rightarrow 0^+$ in the model, the length of each season tends to zero. As a consequence, this limit corresponds to the disappearance of any seasonality in the model and all processes are assumed to occur continuously in time, as, for example, in the Klausmeier model \eqref{eq: Intro Klausmeier local}.
	
	One kernel function satisfying the conditions in Prop. \ref{prop: Difference: limit case} is the Laplacian kernel \eqref{eq:difference equations: simulations: Laplacian}. We define the corresponding asymmetric Laplace kernel by $\phi_1(x) =  Ne^{-a_2x}$ for $x\ge 0$ and $\phi_1(x) = Ne^{(a_2-a_1)x}$ for $x < 0$, where $N = (a_2-a_1)a_2/(2a_2-a_1)$ and $a_2>a_1>0$. The parameter $a_1$ controls the extent of the asymmetry of the kernel function and $a_1=0$ yields the symmetric Laplace kernel \eqref{eq:difference equations: simulations: Laplacian}.  The model with this particular kernel function is studied in some detail in this paper as the Fourier transform of the symmetric Laplacian kernel $\widehat{\phi}(k) = a^2/(a^2+k^2)$ provides a significant simplification in the analysis of pattern onset.}

\begin{proof}[Proof of Proposition \ref{prop: Difference: limit case}]
	{\color{changes}Firstly, we show that
		\begin{align}\label{eq: Difference: consistency proof first step}
		P_T\bm{v}(x,t) = \frac{\bm{v}(x,t+T) - \bm{v}(x,t)}{T} - \Gamma\bm{v}(x,t) - h_1(\bm{v}(x,t)) + O\left(T^2\right).
		\end{align} 
		
		To this end, we define $\phi(x) = \sigma_{\phi}^{-1}\varphi(\sigma_{\phi}^{-1}x)$ and $\phi_1(x) = \widetilde{\sigma}_{\phi_1}^{-1}\varphi_1(\widetilde{\sigma}_{\phi_1}^{-1}x)$ Under the changes of variables $y=x-\sigma_{\phi}z$ and $y=x-\widetilde{\sigma}_{\phi_1}z$, respectively, $P_T \bm{v} = ((P_T\bm{v})_1,(P_T\bm{v})_2)$ satisfies
		\begin{subequations}
			\label{eq:difference: comparison: variable change}
			\begin{align}
			T\left(P_T\bm{v}\right)_1 &= u(x,t+T) - C\int_{-\infty}^\infty \varphi(z) f\left(u\left(x-\sigma_{\phi}z,t \right), w\left(x-\sigma_{\phi}z,t \right) \right) \dif z ,\\ \begin{split}
			T\left(P_T\bm{v}\right)_2 &= w(x,t+T) - D\int_{-\infty}^\infty \varphi_1(z) g\left(u\left(x-\widetilde{\sigma}_{\phi_1}z,t \right), w\left(x-\widetilde{\sigma}_{\phi_1}z,t \right) \right) \dif z. \end{split}
			\end{align}
		\end{subequations}
		
		Due to the parameter setting \eqref{eq:difference: limit scalings}, small values of $T$ correspond to small values of $\sigma_{\phi}$ and $\widetilde{\sigma}_{\phi_1}$ Hence, to investigate the system's behaviour for $T\ll 1$, consider the Taylor expansions of $u(x-\sigma_{\phi}z,t)$, $w(x-\sigma_{\phi}z,t)$, $u(x-\widetilde{\sigma}_{\phi_1}z,t)$ and $w(x-\widetilde{\sigma}_{\phi_1}z,t)$ about $x$, which give
		
		\begin{multline} \label{eq:difference: taylor expansion f}
		f\left(u\left(x-\sigma_{\phi}z,t \right), w\left(x-\sigma_{\phi}z,t \right) \right) \\ = u(x,t)^2w(x,t)+\left(\frac{1}{C}-B\right)u(x,t) \\ -\sigma_{\phi}z \left(u(x,t)^2 w_x(x,t) + \left(\frac{1}{C}-B\right)u_x(x,t) + 2u(x,t)u_x(x,t)w(x,t)\right) \\ +\sigma_{\phi}^2z^2 \left(\frac{1}{2}u(x,t)^2 w_{xx}(x,t) + u_x(x,t)^2w(x,t) + \frac{1}{2}\left(\frac{1}{C}-B\right)u_{xx}(x,t) \right. \\ \left. +u(x,t)u_{xx}(x,t)w(x,t) + 2u(x,t)u_x(x,t)w_x(x,t) \right) + O\left(\sigma_{\phi}^3\right),
		\end{multline}
		and similarly
		\begin{multline}\label{eq:difference: taylor expansion g}
		g\left(u\left(x-\widetilde{\sigma}_{\phi_1}z \right), w\left(x-\widetilde{\sigma}_{\phi_1}z \right) \right) \\ = A - u(x,t)^2w(x,t) + \left(\frac{1}{D} -1\right)w(x,t) \\ - \widetilde{\sigma}_{\phi_1}z \left(-u(x,t)^2w_x(x,t) + \left(\frac{1}{D} -1\right) w_x(x,t) - 2u(x,t)u_x(x,t)w(x,t)\right) \\ + \widetilde{\sigma}_{\phi_1}^2z^2 \left(-u_x(x,t)^2w(x,t) - \frac{1}{2}u(x,t)^2w_{xx}(x,t) + \frac{1}{2}\left(\frac{1}{D} -1\right) w_{xx}(x,t) \right. \\ \left.  -u(x,t)u_{xx}(x,t)w(x,t) - 2u(x,t)u_x(x,t)w_x(x,t) \right) + O\left(\widetilde{\sigma}_{\phi_1}^3\right),
		\end{multline}
		where the subscripts of $u$ and $w$ denote partial differentiation. Substitution of this into \eqref{eq:difference: comparison: variable change} and term-wise integration using Watson's Lemma (e.g. \cite{Miller2006}) gives 
		\begin{multline*}
		T\left(P_T\bm{v}\right)_1 = u(x,t+T) -  C\left(u(x,t)^2w(x,t) + \left(\frac{1}{C}-B\right)u(x,t) \right. \\ \left. + \left(u(x,t)^2w_{xx}(x,t) + 2(u_{x}(x,t))^2w(x,t) + \left(\frac{1}{C}-B\right)u_{xx}(x,t) \right. \right.\\ \left. \left.+2u(x,t)u_{xx}(x,t)w(x,t) + 4u(x,t)u_{x}(x,t)w_{x}(x,t) \right) \sigma_{\phi}^2 \int_{-\infty}^\infty \varphi(z)z^2\dif z + O\left(\sigma_{\phi}^3\right)\right),  
		\end{multline*}
		and
		\begin{multline*}
		T\left(P_T\bm{v}\right)_2 = w(x,t+T) \\ -  D\left(  2\left(A-u(x,t)^2w(x,t) +\left(\frac{1}{D}-1\right) w(x,t)\right) \right.  \left. \int_{-\infty}^\infty \varphi_1(z)\dif z \right. \\ \left.  +  \left(u(x,t)^2w_{x}(x,t)-\left(\frac{1}{D} -1\right)w_{x}(x,t) +2u(x,t)u_{x}(x,t)w(x,t)\right) \widetilde{\sigma}_{\phi_1}\int_{-\infty}^\infty \varphi_1(z)z\dif z   \right. \\ 
		\left. + \left(-2(u_{x}(x,t))^2w(x,t) - u(x,t)^2w_{xx}(x,t) +\left(\frac{1}{D}-1\right) w_{xx}(x,t) \right.\right. \\ \left.\left. -2u(x,t)u_{xx}(x,t)w(x,t) -4u(x,t)u_{x}(x,t)w_{x}(x,t) \right) \frac{\widetilde{\sigma}_{\phi_1}^2}{2}\int_{-\infty}^\infty \varphi_1(z)z^2\dif z   + O\left(\widetilde{\sigma}_{\phi_1}^3\right) \right).
		\end{multline*}
		Using that $\varphi(x) = \sigma_{\phi}\phi(\sigma_{\phi}x)$, $\varphi_1(x) = \widetilde{\sigma}_{\phi_1}\phi_1(\widetilde{\sigma}_{\phi_1}x)$, and the definition of the moments of a probability distribution give
		\begin{multline*}
		T\left(P_T\bm{v}\right)_1 = u(x,t+T) -   C \left( u(x,t)^2w(x,t)+ \left(\frac{1}{C}-B\right) u(x,t) + \frac{\sigma_{\phi}^2}{2C}u_{xx}(x,t) \right.\\ \left. +  \sigma_{\phi}^2 \left(\frac{1}{2}u(x,t)^2w_{xx}(x,t) + (u_{x}(x,t))^2w(x,t) - \frac{1}{2}Bu_{xx}(x,t) \right.\right. \\ \left.\left. +u(x,t)u_{xx}(x,t)w(x,t) +2u(x,t)u_{x}(x,t)w_{x}(x,t)\right) + O\left(\sigma_{\phi}^3\right)\right), 
		\end{multline*}
		and
		\begin{multline*}
		T\left(P_T\bm{v}\right)_2 = w(x,t+T) - D\left( A-u(x,t)^2w(x,t) +\left(\frac{1}{D}-1 \right) w(x,t) \right. \\ \left. -\frac{\mu_{\phi_1}}{D} w_{x}(x,t) +\frac{\widetilde{\sigma}_{\phi_1}^2}{2D}w_{xx}(x,t) +\mu_{\phi_1}\left(u(x,t)^2w_{x}(x,t)+w_{x}(x,t)  \right.\right. \\ \left. \left. +2u(x,t)u_{x}(x,t)w(x,t) \right) + \widetilde{\sigma}_{\phi_1}^2\left(-(u_{x}(x,t))^2w(x,t)-\frac{1}{2}u(x,t)^2w_{xx}(x,t) \right.\right. \\ \left.\left.-\frac{1}{2}w_{xx}(x,t) -u(x,t)u_{xx}(x,t)w(x,t) -2u(x,t)u_{x}(x,t)w_{x}(x,t) \right) + O\left(\widetilde{\sigma}_{\phi_1}^3\right) \right).
		\end{multline*}
		
		Applying \eqref{eq:difference: limit scalings} yields
		\begin{multline*}
		T\left(P_T\bm{v}\right)_1 = u(x,t+T)   \\ -\left(u(x,t) + T\left(u(x,t)^2w(x,t)-Bu(x,t)+u_{xx}(x,t)\right) \right)+ O\left(T^2\right),
		\end{multline*}
		and
		\begin{multline*}
		T\left(P_T\bm{v}\right)_2 = w(x,t+T) \\ -\left(w(x,t) + T\left(A-u(x,t)^2w(x,t)-w(x,t)+\nu w_x(x,t)+dw_{xx}(x,t) \right) \right) + O\left(T^2\right),
		\end{multline*}
		which shows \eqref{eq: Difference: consistency proof first step}.
		
		The Taylor expansions $u(x,t+T) = u(x,t) + T u_t (x,t) +O(T^2)$ and $w(x,t+T) = w(x,t) + T w_t (x,t)+O(T^2)$ yield
		\begin{align*}
		P_T\bm{v}(x,t) = \frac{\partial \bm{v}}{\partial t}(x,t) - \Gamma\bm{v}(x,t) - h(\bm{v}(x,t)) +O\left(T^2\right),
		\end{align*}
		and thus
		\begin{align*}
		P\bm{v} - P_T\bm{v} = O\left(T^2\right),
		\end{align*}
		which tends to zero as $T\rightarrow 0$.
	}
	
\end{proof}


\section{Linear Stability Analysis}\label{sec: integrodifference: linstab}
A common approach to study the onset of spatial patterns in a model is linear stability analysis. Spatial patterns occur if a steady state that is stable to spatially homogeneous perturbations becomes unstable if a spatially heterogeneous perturbation is introduced.
In this section we show that such a linear stability analysis of the integrodifference model \eqref{eq: Klausmeier integrodifference} on flat ground with the Laplacian kernels in the parameter setting \eqref{eq:difference: limit scalings} yields a condition for pattern onset that is equivalent to the corresponding condition for the local Klausmeier model \eqref{eq: Intro Klausmeier local}. This implies that pattern onset is independent of the parameter $T$, the temporal separation of seed dispersal events. 


The steady states of \eqref{eq: Klausmeier integrodifference} are identical with those of the Klausmeier models \eqref{eq: Intro Klausmeier local} and \eqref{eq: Intro Klausmeier nonlocal}, i.e.
\begin{align*}
(\overline{u}_1,\overline{w}_1) &= \left(0,A\right), \quad 
(\overline{u}_2,\overline{w}_2) = \left(\frac{2B}{A-\sqrt{A^2-4B^2}},\frac{A-\sqrt{A^2-4B^2}}{2}\right), \\
(\overline{u}_3,\overline{w}_3) &= \left(\frac{2B}{A+\sqrt{A^2-4B^2}},\frac{A+\sqrt{A^2-4B^2}}{2}\right).
\end{align*} 
Existence of $(\overline{u}_2,\overline{w}_2)$ and $(\overline{u}_3,\overline{w}_3)$ requires $A>A_{\min}:=2B$. The steady states are independent of $C$, $D$ and the dispersal widths $a$, $a_1$ and $a_2$ and are thus independent of frequency changes to the temporal intermittency when using the scalings \eqref{eq:difference: limit scalings}.  For the Klausmeier models $(\overline{u}_1,\overline{w}_1)$ and $(\overline{u}_2,\overline{w}_2)$ are stable to spatially homogeneous perturbations, while $(\overline{u}_3,\overline{w}_3)$ is unstable to spatially homogeneous perturbations in the biologically relevant parameter region $B<2$ \cite{Klausmeier1999, Sherratt2005, Eigentler2018nonlocalKlausmeier}. Preservation of this structure of the steady states in the integrodifference model \eqref{eq: Klausmeier integrodifference} is only achieved in a certain parameter region.

\begin{prop}\label{prop: Difference: LinStab: stability homogeneous perturbations}
	If 	
	\begin{align}\label{eq:difference equations: time only stability region for all A,B}
	D= \ell\overline{D}, \ \ell< 1, \quad C=\frac{\ell_1\overline{D}}{B(m-\ell_1\overline{D})}, \ m > 2, \ell_1<1,
	\end{align}
	where 
	\begin{align}\label{eq:difference equations: Dbar}
	\overline{D} = \frac{2\left(A^2-A\sqrt{A^2-4B^2} - 2B^2\right)}{A^2-A\sqrt{A^2-4B^2}},
	\end{align}
	then $(\overline{u}_1,\overline{w}_1)$ and $(\overline{u}_2,\overline{w}_2)$ are stable to spatially homogeneous perturbations, and $(\overline{u}_3,\overline{w}_3)$ is unstable to spatially homogeneous perturbations.
\end{prop}

This condition is sufficient but not necessary. Outside this region further restrictions on the rainfall parameter $A$ can be imposed to guarantee conservation of the steady state structure. In the limiting case \eqref{eq:difference: limit scalings} such a restriction on the rainfall parameter cannot be avoided. The following condition ensures that \eqref{eq:difference equations: time only stability region for all A,B} holds in the limiting case \eqref{eq:difference: limit scalings}.
\begin{cor}\label{cor: Difference: LinStab: parameter region providing steady state structure in limit case}
	If
	\begin{align}\label{eq: Difference: LinStab: parameter region limit case}
	A^2 < A_+^2:=\min\left\{  \frac{4B^2}{(2-T)T}, \frac{B(BT+1)^2}{T}\right\}, \quad T <\frac{1}{2}, \quad B<2,
	\end{align}
	in \eqref{eq: Klausmeier integrodifference} with $C=D=T$, then $(\overline{u}_1,\overline{w}_1)$ and $(\overline{u}_2,\overline{w}_2)$ are stable to spatially homogeneous perturbations, and $(\overline{u}_3,\overline{w}_3)$ is unstable to spatially homogeneous perturbations.
\end{cor}
In the limit $T\rightarrow 0^+$ this becomes the whole $A$-$B$ parameter region considered for the continuous-time Klausmeier models, providing a reasonable framework for a comparison of the two models. The upper bounds on $T$ and $A$ do, however, introduce a significant restriction on the model as no arbitrarily large time between dispersal events or large precipitation volumes $A$ can be considered. In this, as well as the parameter region given by \eqref{eq:difference equations: time only stability region for all A,B}, the plant density $u_n(x)$ and the water density $w_n(x)$ remain positive for initial conditions close to the steady states. This is sufficient for the linear stability analysis and simulations that follow.
In the parameter region in which $(\overline{u}_2,\overline{w}_2)$ is unstable, four different behaviours of the system's solution can be observed; (i) convergence to the desert steady state, (ii) divergence, (iii) a chaotic solution or (iv) a periodic solution for which period doubling occurs as $T$ is increased. However, these different behaviours can yield negative densities of the system's quantities and are thus not considered further in this paper.

Spatial patterns of \eqref{eq: Klausmeier integrodifference} arise if the steady state $(\overline{u}, \overline{w}):=(\overline{u}_2, \overline{w}_2)$, which is stable to spatially homogeneous perturbations, becomes unstable if a spatially heterogeneous perturbation is introduced. 

\begin{prop}\label{prop: Difference: LinStab: stability heterogeneous perturbations general}
	The steady state $(\overline{u}, \overline{w})$ is stable to spatially heterogeneous perturbations if $|\lambda(k)|<1$ for both eigenvalues of the Jacobian
	\begin{align}\label{eq:difference: linear stability Jacobian}
	J = \mattwo{C\widehat{\phi}(k)\alpha}{C\widehat{\phi}(k)\beta}{D\widehat{\phi_1}(k)\gamma}{D\widehat{\phi_1}(k)\delta},
	\end{align}
	for all $k>0$, where
	\begin{align}\label{eq: difference: LinStab: linearisation coefficients}
	\begin{split}
	\alpha & = f_u(\overline{u},\overline{w}) = \frac{BC+1}{C}, \\ \beta &= f_w(\overline{u},\overline{w}) = \frac{4B^2}{\left(A-\sqrt{A^2-4B^2}\right)^2}, \quad
	\gamma = g_u(\overline{u},\overline{w}) = -2B, \\ \delta &= g_w(\overline{u},\overline{w}) = -\frac{2\left(A^2D-AD\sqrt{A^2-4B^2} -A^2+A\sqrt{A^2-4B^2} +2B^2\right)}{D\left(A-\sqrt{A^2-4B^2}\right)^2}.
	\end{split}
	\end{align}
\end{prop}

Due to the asymmetry of $\phi_1$ some of the entries of the Jacobian \eqref{eq:difference: linear stability Jacobian} are complex-valued. A significant simplification can therefore be achieved by considering the integrodifference model \eqref{eq: Klausmeier integrodifference} on flat ground. This corresponds to $a_1 =0$ in $\phi_1$. As a consequence, the Jury conditions (see e.g. \cite{Murray1989}) can be used to determine the steady state's stability to spatially heterogeneous perturbations.  To study this in more detail, and in particular to show that the model does not provide information on effects the temporal separation of seed dispersal events, we focus on the limiting case \eqref{eq:difference: limit scalings} and the Laplacian kernel \eqref{eq:difference equations: simulations: Laplacian}. 

\begin{prop}\label{prop: Difference: LinStab: sufficient condition for patterns on flat ground}
	The steady state $(\overline{u}, \overline{w})$ of the integrodifference model \eqref{eq: Klausmeier integrodifference} under the scalings \eqref{eq:difference: limit scalings} on flat ground with the Laplacian kernels \eqref{eq:difference equations: simulations: Laplacian} is unstable to spatially heterogeneous perturbations if
	\begin{align}\label{eq:difference: linear stability: flat ground jury2}
	1+\det(J)- |\tr(J)| <0, \quad \text{for some } k>0,
	\end{align}
	where J is the Jacobian given in Proposition \ref{prop: Difference: LinStab: stability heterogeneous perturbations general} with $a_1=0$.
\end{prop}


In other words, Proposition \ref{prop: Difference: LinStab: sufficient condition for patterns on flat ground} provides a sufficient condition for spatial patterns to occur. The following proposition shows that \eqref{eq:difference: linear stability: flat ground jury2} is equivalent to the stability condition \eqref{eq: Models: Klausmeier: flat ground diffusion threshold} of $(\overline{u},\overline{w})$ in the local Klausmeier model. In other words, a diffusion driven instability causes the occurrence of spatial patterns in the integrodifference model, i.e. given a level of rainfall $A$, an instability occurs for $d>d_c(A,B)$, where $d_c(A,B)$ is given in \eqref{eq: Models: Klausmeier: flat ground diffusion threshold}. 

\begin{prop}\label{prop: Difference: LinStab: diffusion threshold flat ground}
	{\color{changes}The steady state $(\overline{u}, \overline{w})$ of the integrodifference model \eqref{eq: Klausmeier integrodifference} under the scalings \eqref{eq:difference: limit scalings} on flat ground with the Laplacian kernels \eqref{eq:difference equations: simulations: Laplacian} is unstable to spatially heterogeneous perturbations if $d>d_c(A,B)$, where the threshold $d_c$ is identical with the corresponding threshold \eqref{eq: Models: Klausmeier: flat ground diffusion threshold} for the local Klausmeier model.}
\end{prop}


The condition's independence of $T$ yields that the integrodifference model does not provide any information on the effects of the temporal separation of seed dispersal events on the onset of spatial patterns. The equivalence of the condition to that of the local Klausmeier model follows directly from the condition's independence of $T$ and Proposition \ref{prop: Difference: limit case}, which shows that the integrodifference model converges to the local Klausmeier model as $T\rightarrow 0^+$. Thus for sufficiently small values of $T$, Proposition \ref{prop: Difference: LinStab: diffusion threshold flat ground} does indeed provide the exact same information as the diffusion threshold obtained for the local Klausmeier model. For larger $T$ the model does not provide any information on the transition between uniform and patterned vegetation as the decrease in the upper bound $A_+$ on the rainfall parameter reduces the size of the rainfall interval for which the derivation of $d_c$ is valid.

\begin{proof}[Proof of Proposition \ref{prop: Difference: LinStab: stability homogeneous perturbations}]
	Stability of a steady state $(\overline{u},\overline{w})$ is determined by the Jury conditions applied to the Jacobian
	\begin{align*}
	J = \mattwo{C(2\overline{u}\overline{w}-B)+1}{C\overline{u}^2}{-2D\overline{u}\overline{w}}{-D(\overline{u}^2+1)+1}.
	\end{align*}
	The steady state $(\overline{u}_3,\overline{w}_3)$ is unstable in the whole parameter region, because 
	\begin{align*}
	1+\det(J)-|\tr(J)| = -\frac{2BCD\left(A^2+A\sqrt{A^2-4B^2}-4B^2 \right)}{\left(A+\sqrt{A^2-4B^2}\right)^2} <0.
	\end{align*}
	The desert steady state $(\overline{u}_1,\overline{w}_1)$ is monotonically stable if $C<B^{-1}$ and $D<1$. If $1<D<2$ or $B^{-1}<C<2B^{-1}$ it is still stable but solutions are oscillating about $(0,A)$, which is biologically impossible. Finally, the Jury conditions yield that $(\overline{u}_2,\overline{w}_2)$ is stable to spatially homogeneous perturbations if $\min\{\overline{C_2}, \overline{C_3}\} < C < \overline{C_1}$, where
	\begin{align*}
	\overline{C_1} &= \frac{AD\left(A-\sqrt{A^2-4B^2}\right)}{B\left((D-1)A\left(\sqrt{A^2-4B^2}-A\right)+2B^2(2D-1) \right)}, \\
	\overline{C_2} &= \frac{2\left((D-2)\left(A\sqrt{A^2-4B^2}-A^2\right)-4B^2\right)}{B\left((D-2)\left(A^2-A\sqrt{A^2-4B^2} \right)-4B^2(D-1) \right)}, \\
	\overline{C_3} &= \frac{(D-2)\left(A^2-A\sqrt{A^2-4B^2}\right)+4B^2}{B\left(A^2-A\sqrt{A^2-4B^2}-2B^2\right)}.
	\end{align*}
	Combined, this gives that the steady state structure of the continuous time model is preserved if
	\begin{align}\label{eq:difference equations: time only stability region}
	D<1 \quad \text{and} \quad \max\left\{0,\min\left\{\overline{C_2},\overline{C_3} \right\} \right\} < C < \min\left\{\frac{1}{B}, \overline{C_1} \right\}.
	\end{align}

	If $D > 1/2$, then $\min\left\{1/B, \overline{C_1} \right\} =  1/B$, because
	\begin{align*}
	\overline{C_1} - \frac{1}{B} = -\frac{2\left(D-\frac{1}{2}\right)\left(A^2-A\sqrt{A^2-4B^2}-2B^2\right)}{B\left(\left(D-\frac{1}{2} \right)\left(A^2-A\sqrt{A^2-4B^2}-4B^2\right)-\left(A^2-A\sqrt{A^2-4B^2}\right)\right)}>0,
	\end{align*}
	since $A^2-A\sqrt{A^2-4B^2}-2B^2>0$ and $A^2-A\sqrt{A^2-4B^2}-4B^2<0$.
	Similarly, if $D < 1/2$, then $\min\left\{1/B, \overline{C_1} \right\} = \overline{C_1}$. Further, if $D < \overline{D}$ (defined in \eqref{eq:difference equations: Dbar}), then $\max\{0,\min\{\overline{C_2},\overline{C_3} \} \} = 0$ and similarly, if $D > \overline{D}$, then $\max\{0,\min\{\overline{C_2},\overline{C_3} \} \} = \overline{C_2}$.
	
	Hence, \eqref{eq:difference equations: time only stability region} can be simplified by splitting it into different parameter regions. It becomes (i) $C<\overline{C_1}$ if $D<1/2$ and $D<\overline{D}$, (ii) $\overline{C_2}<C<\overline{C_1}$ if $D<1/2$ and $\overline{D}<D<1$, (iii) $C<1/B$ if $1/2<D<1$ and $D<\overline{D}$ and (iv) $\overline{C_2}<C<1/B$ if  $1/2<D<1$ and $\overline{D}<D<1$. 
	This classification is used below to show that if $C$ and $D$ are defined as in \eqref{eq:difference equations: time only stability region for all A,B}, then \eqref{eq:difference equations: time only stability region} is satisfied in the whole parameter plane that is considered in the continuous-time PDE models ($A>2B$, $B<2$). To show this it is sufficient to show that (i) and (iii) are satisfied because $\ell<1$. For case (iii) note that
	\begin{align*}
	C=\frac{\ell_1D}{B(m-\ell_1D)} <\frac{1}{B} \iff \ell_1D<\frac{m}{2},
	\end{align*}
	which is satisfied since $\ell_1D<1$ and $m>2$. For case (i) note that
	\begin{align*}
	C=\frac{\ell_1D}{B(m-\ell_1D)} < \overline{C_1} \iff \ell_1D < \frac{2B^2+(m-1)\left(A^2-A\sqrt{A^2+4B^2}\right)}{4B^2} :=\overline{\overline{D}}.
	\end{align*}
	This is always satisfied because $\ell_1D<1$ and
	\begin{align*}
	\overline{\overline{D}} >1 \iff m>\frac{2B^2}{A^2-A\sqrt{A^2+4B^2}} +1 :=\overline{m},
	\end{align*}
	which holds true since $m>2$ and 
	\begin{align*}
	\overline{m} <2 \iff A^2-A\sqrt{A^2-2B^2}-2B^2 > 0,
	\end{align*}
	which is clearly satisfied.
	
\end{proof}

\begin{proof}[Proof of Proposition \ref{prop: Difference: LinStab: stability heterogeneous perturbations general}]
	
	Linearisation of the model \eqref{eq: Klausmeier integrodifference} about the steady state $(\overline{u}, \overline{w})$ gives $u_{n+1}(x) = C  \phi(\cdot) (\alpha u_n(\cdot) + \beta w_n(\cdot))$ and $w_{n+1}(x) = D  \phi_1(\cdot)(\gamma u_n(\cdot) +  \delta w_n(\cdot))$. Taking the Fourier transform of both equations yields $\widehat{u}_{n+1}(k) = C \widehat{\phi}(k) (\alpha\widehat{u}_n(k) + \beta \widehat{w}_n(k))$ and $\widehat{w}_{n+1}(k) = D \widehat{\phi_1}(k)(\gamma \widehat{u}_n(k) + \delta \widehat{w}_n(k))$, where $\widehat{\phi}$ and $\widehat{\phi_1}$ denote the Fourier transforms of the kernels $\phi$, and $\phi_1$, respectively. Under the assumption that $\widehat{u}_n(k)$ and $\widehat{w}_n(k)$ are proportional to $\lambda^n\tilde{u}(k)$ and $\lambda^n\tilde{w}(k)$, respectively, where $\lambda \in \C$ denotes the growth rate, the system becomes $\lambda \tilde{u}(k) = C \widehat{\phi}(k) (\alpha\tilde{u}(k) + \beta \tilde{w}(k))$ and $\lambda\tilde{w}(k) = D \widehat{\phi_1}(k)(\gamma \tilde{u}(k) + \delta \tilde{w}(k))$, i.e. $\lambda$ is an eigenvalue of the Jacobian $J$.
\end{proof}

\begin{proof}[Proof of Proposition \ref{prop: Difference: LinStab: sufficient condition for patterns on flat ground}]
	For an instability to occur, at least on of the Jury conditions $\det (J) <1$ and $1+\det (J) - |\tr (J)|>0$ needs to be violated for some wavenumber $k>0$. The former condition is satisfied for all $k>0$. To show this, note that that  $\max \{\det(J) -1\}$ is at $k=0$ because
	\begin{align}\label{eq:difference: linear stability: flat ground jury 1 det - 1}
	\det(J) - 1 = \frac{\alpha_4 k^4 + \alpha_2 k^2 + \alpha_0}{ \left(d T k^2+1\right)\left(A-\sqrt{A^2-4B^2}\right)^2\left(T k^2 +1\right)},
	\end{align}
	where
	\begin{align*}
	\alpha_4 &= 2d T^2 \left(-A^2+A\sqrt{A^2-4B^2}+2B^2\right), \\
	\alpha_2 &= -2T\left(A^2-A\sqrt{A^2-4B^2}-2B^2\right)(d+1), \\
	\alpha_0 &= 2T\left(\left(\frac{1}{2}B-1\right)\left(A^2-A\sqrt{A^2-4B^2}\right) \right. \\ &\left. + \left(\frac{1}{2}B-TB\right)\left(A^2-A\sqrt{A^2-4B^2}-4B^2\right) \right).
	\end{align*}
	The denominator of \eqref{eq:difference: linear stability: flat ground jury 1 det - 1} is clearly positive and increasing for $k>0$. Since further $\alpha_4 <0$ and $\alpha_2<0$, the numerator and thus the whole of \eqref{eq:difference: linear stability: flat ground jury 1 det - 1} is decreasing for $k>0$ and it attains its maximum at $k=0$. The negativity of  \eqref{eq:difference: linear stability: flat ground jury 1 det - 1} then follows from that of $\alpha_0$ which follows from $B<2$ and $T<1/2$.
\end{proof}


\begin{proof}[Proof of Proposition \ref{prop: Difference: LinStab: diffusion threshold flat ground}]
	{\color{changes}Firstly, we note that $\partial d_c / \partial A \ge 0$ for all $A\ge 2B$. Hence, $d_c$ attains its minimum on $A=2B$, on which it simplifies to $d_c = 2/B$. Since $B<2$, $d_c>1$.
		Next, we} show that $\tr(J)>0$. To {\color{changes}do} this, note that
	\begin{align*}
	\tr(J) = \frac{\beta_2k^2+\beta_0}{ \left(d T k^2+1\right)\left(A-\sqrt{A^2-4B^2}\right)^2\left(T k^2 +1\right)} >0,
	\end{align*} 
	for all $k>0$, where
	\begin{align*}
	\beta_2 &= 2\left(A^2-A\sqrt{A^2-4B^2}-2B^2\right)\left(BT^2d+T+Td\right)-2T^2\left(A^2-A\sqrt{A^2-4B^2}\right), \\ 
	\beta_0 &= 2\left(BT-T+2\right) \left(A^2-A\sqrt{A^2-4B^2}-2B^2\right). 
	\end{align*}
	The denominator is clearly positive and thus the condition for positivity of $\tr(J)$ is $\beta_2k^2+\beta_0>0$. The left hand side of this is decreasing in $A$ since $A^2-A\sqrt{A^2-4B^2}$ is decreasing in $A$ and the assumptions on $B$ and $d$, and thus obtains its minimum at $A=A^+$, where $A^+$ is given in \eqref{eq: Difference: LinStab: parameter region limit case}. If $B<1/(2-T)$, then $A^+ = 4B^2/((2-T)T)$ and 
	\begin{align*}
	\tr(J)\left(\sqrt{A^+}\right) >0 \Longleftrightarrow k^2 > \frac{B}{1-d-BTd},
	\end{align*}
	{\color{changes} since} $d>1$. The right hand side is negative and thus $\min(\tr(J)) >0$ for $B<1/(2-T)$. If $B>1/(2-T)$, then $A^+ = (BT+1)^2B/T$ and 
	\begin{align*}
	\tr(J)\left(\sqrt{A^+}\right) >0 \Longleftrightarrow k^2 > -\frac{TB^2+(2-T)B-1}{B^2T^2d+\left((d+1)T-T^2\right)B-T},
	\end{align*}
	{\color{changes} since} $d>1$. Negativity of the right hand side follows from the lower bound on $B$ and thus  $\min(\tr(J)) >0$ for all $B<2$. This shows that $\tr(J)>0$. The stability condition \eqref{eq:difference: linear stability: flat ground jury2} thus becomes $1+\det(J)-\tr(J) < 0 \Longleftrightarrow \gamma_4k^4+\gamma_2k^2+\gamma_0 <0$, where $\gamma_4 = d(A^2-A\sqrt{A^2-4B^2}-2B^2)$, $\gamma_2 = (A^2-A\sqrt{A^2-4B^2})(1-Bd)+2B^3d$ and	$\gamma_0 = B(A^2-4B^2)$. This condition and thus its minimum $-\gamma_2^2/(4\gamma_4)+\gamma_0$ is independent of $T$. Determining the locus at which the minimum changes sign gives the threshold $d_c(A,B)$.
\end{proof}


\section{Simulations}\label{sec:difference: simulations}
The preceding linear stability analysis relies on the use of the Laplace kernel. For other kernel functions whose Fourier transforms do not provide such a simplification numerical simulations of the model are considered to investigate the onset of patterns. In particular, this allows us to make comparisons between different dispersal kernels, similar to the analysis performed for the nonlocal model in \cite{Eigentler2018nonlocalKlausmeier}. These show that both wide plant dispersal kernels and narrow water dispersal kernels inhibit the formation of patterns. Finally in this section, we show that as for the nonlocal Klausmeier model, the kind of decay of the plant dispersal kernel at infinity is also important.

Simulations are performed on the space domain $[-x_{\max}, x_{\max}]$ centred at $x=0$. This domain is discretised into $M$ equidistant points $x_1,\dots,x_M$ with $-x_{\max} = x_1 < x_2 < \dots < x_M = x_{\max}$ such that $\Delta x = x_2-x_1 = \dots =x_M - x_{M-1}$. On flat ground \eqref{eq: Klausmeier integrodifference} then becomes
\begin{subequations} \label{difference equations: simulations: discretised system}
	\begin{align}
	u_{n+1}(x_k) &= C \Delta x \left(\phi \ast f_n \right)_k, \label{difference equations: simulations: discretised system plants} \\
	w_{n+1}(x_k) &= D\Delta x \left(\phi_1 \ast g_n \right)_k, \label{difference equations: simulations: discretised system water}
	\end{align}
\end{subequations}
where $\phi, \phi_1$ denote the vectors consisting of the elements obtained by evaluating the corresponding function at each mesh point, $f_n, g_n$ denote the vectors consisting of the elements obtained by evaluating the corresponding function at each $(u_n(x_k), w_n(x_k))$ and $z_1 \ast z_2$ denotes the discrete convolution of two vectors $z_1$ and $z_2$. The convolution terms in \eqref{difference equations: simulations: discretised system plants} and \eqref{difference equations: simulations: discretised system water} are obtained by using the convolution theorem and the fast Fourier transform, providing a significant simplification as this reduces the number of operations required to obtain the convolution from $O(M^2)$ to $O(M\log(M))$ (see e.g. \cite{Cooley1969}).

To mimic the infinite domain used for the linear stability analysis (Section \ref{sec: integrodifference: linstab}), we  define the initial condition of the system as follows; on a subdomain $[-x_{\operatorname{sub}}, x_{\operatorname{sub}} ]$ centred at $x=0$ of the domain $[-x_{\max}, x_{\max}]$ considered in the simulation the initial condition is a random perturbation of the steady state $(\overline{u}, \overline{w})$, while on the rest of the domain the densities are initially set to equal the densities of the steady state $(\overline{u}, \overline{w})$. In other words, $u_0(x_k) = \overline{u} + \delta(x_k)$ and $w_0(x_k) = \overline{w} + \varepsilon(x_k)$ for $x_k \in [-x_{\operatorname{sub}}, x_{\operatorname{sub}} ]$, where $\|\delta\|_\infty < 0.1\overline{u} $ and $\|\varepsilon\|_\infty < 0.1\overline{w}$ and $u_0(x_k) = \overline{u}$ and $w_0(x_k) = \overline{w}$ for $x_k \notin [-x_{\operatorname{sub}}, x_{\operatorname{sub}} ]$. The size of the outer domain is chosen large enough so that any boundary conditions (which are set to be periodic) that are imposed on $[-x_{\max},x_{\max}]$ do not affect the solution in the subdomain in the finite time that is considered in the simulation. Figure \ref{fig:difference equations: simulations: example simulation flat} shows a typical patterned solution obtained by these simulations.

Based on the amplitude of the oscillation relative to the steady state of the solutions obtained by the simulations we set up a scheme to determine the critical rainfall level $A_{\max}$ below which pattern onset occurs. Doing this allows us to investigate how certain changes of parameters and kernel functions affect the onset of patterns. Due to the random perturbation of the initial state of the system, all simulation results shown below are the averages taken over 100 simulations. For the symmetric dispersal kernels $\phi$ and $\phi_1$ we consider the Laplacian \eqref{eq:difference equations: simulations: Laplacian}, the Gaussian
\begin{align}\label{eq:difference equations: simulations: Gaussian}
\phi_g(x) = \frac{a_g}{\sqrt{\pi}}e^{-a_g^2x^2}, \quad a>0, x\in\R,
\end{align}
and the power law distribution
\begin{align}\label{eq:difference equations: simulations: Power Law}
\phi_p(x) = \frac{(b-1)a_p}{2\left(1+a_p|x|\right)^b}, \quad a>0, b >3, x\in\R.
\end{align}

\begin{figure}
	\begin{minipage}{0.49\textwidth}
		\centering
		\centering
		\includegraphics[width=0.9\textwidth]{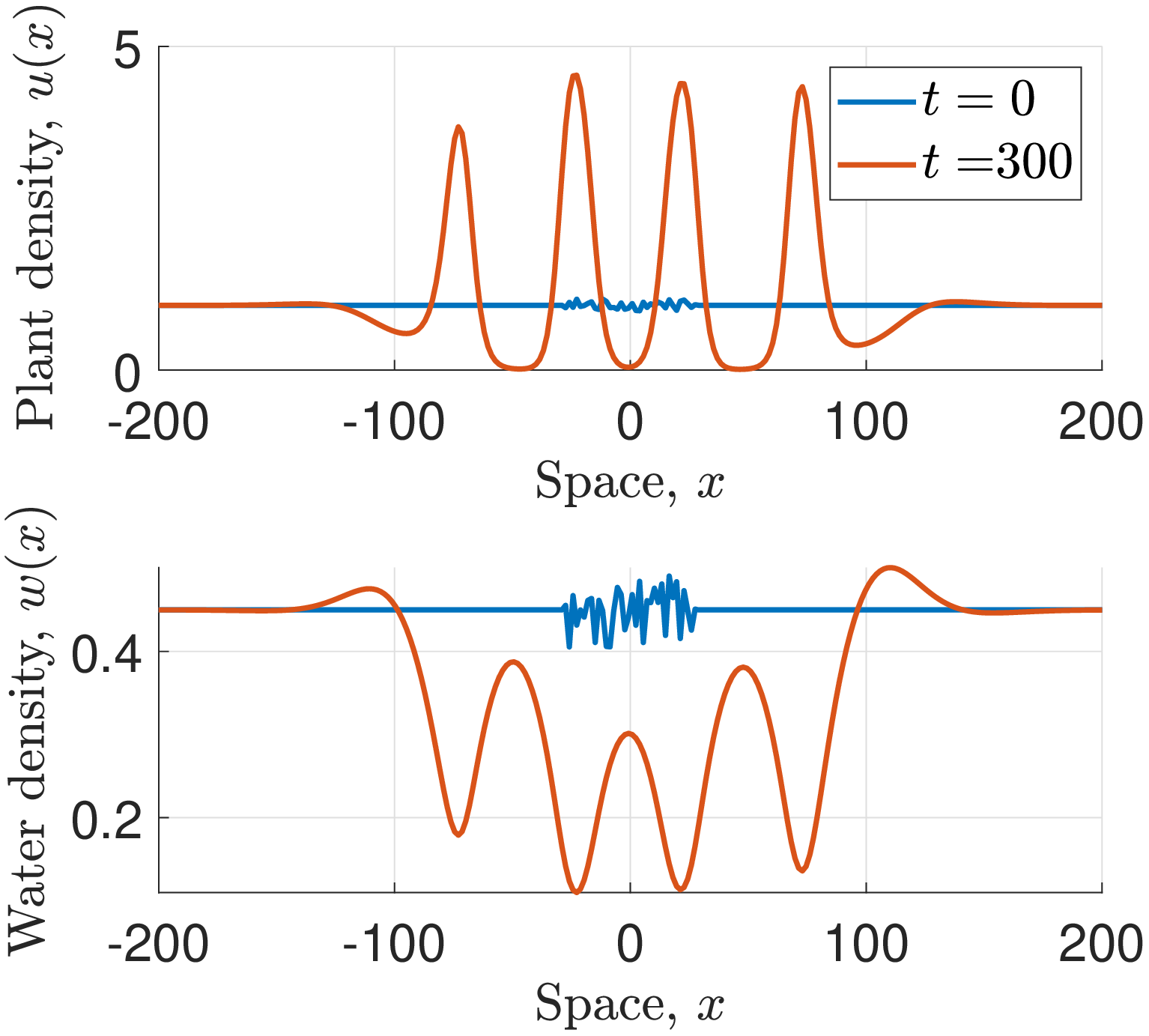}
		\caption{{\color{changes}Simulation of the integrodifference model. This figure shows a patterned solution obtained by simulating the integrodifference model on flat ground. The kernels used in these simulations are the symmetric Laplacian kernels, respectively. The parameter setting \eqref{eq: Difference: LinStab: parameter region limit case} with $T=0.1$ is used in the simulation. The other parameters are $A=0.9$, $B=0.45$ and $d=500$.}}\label{fig:difference equations: simulations: example simulation flat}
	\end{minipage}
	\begin{minipage}{0.01\textwidth}
		\
	\end{minipage}
	\begin{minipage}{0.49\textwidth}
		\centering
		\includegraphics[width=0.9\textwidth]{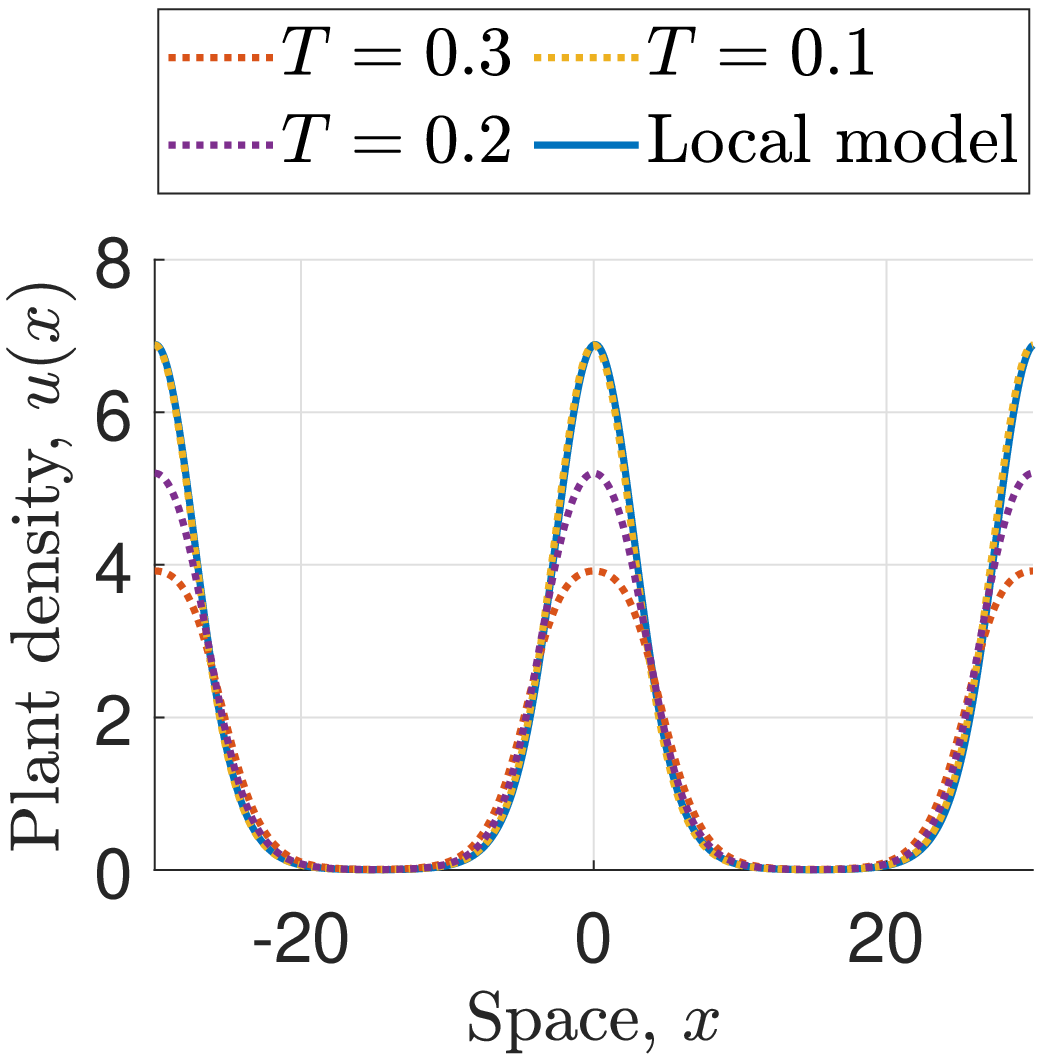}
		\caption{{\color{changes}Convergence of solutions. This figure visualises the convergence of solutions to the local PDE model \eqref{eq: Intro Klausmeier local} as $T\rightarrow 0^+$, to complement the consistency result presented in Prop. \ref{prop: Difference: limit case}. Solutions of the integrodifference model \eqref{eq: Klausmeier integrodifference} are shown for $T=0.3$, $T=0.2$ and $T=0.1$ and are compared with the solution of the local Klausmeier PDE model \eqref{eq: Intro Klausmeier local}. Note that unlike in Fig. \ref{fig:difference equations: simulations: example simulation flat}, the spatial domain is chosen to be small to impose the same wavelength restrictions on both models to aid the visualisation of the convergence. }}\label{fig:difference equations: simulations: convergence simulation}
		
	\end{minipage}
\end{figure}

We base our comparison on the kernels' standard deviations, which are given by $\sigma_\phi = \sqrt{2}/a$ for the Laplacian kernel \eqref{eq:difference equations: simulations: Laplacian}, $\sigma_{\phi_g} = 1/(\sqrt{2}\,a_g)$ for the Gaussian kernel \eqref{eq:difference equations: simulations: Gaussian} and $\sigma_{\phi_p} = \sqrt{2}/(\sqrt{b^2-5b+6}\,a_p)$ for the power law kernel \eqref{eq:difference equations: simulations: Power Law} provided $b>3$. It is perfectly reasonable to perform simulations with kernels of infinite standard deviation (e.g. $b<3$ in the power law kernel) but in the interest of comparing results for the kernels based on their standard deviation we consider only $b=3.1$ and $b=4$.

To investigate the model's behaviour under changes to the dispersal kernels $\phi$ and $\phi_1$, we start by considering simultaneous changes in the kernel functions $\phi$, and $\phi_1$. The comparison between the kernel functions is based on the standard deviation of the plant dispersal kernel $\phi$ and the width of the water dispersal kernel $\phi_1$ is set to $a_2 = 0.1a$ to obtain a ratio similar to that of the standard deviations under the scalings \eqref{eq:difference: limit scalings}, which corresponds to the large value of the diffusion parameter $d$ in the PDE and integro-PDE models. Figure \ref{fig:difference equations: simulations: Amax all kernels flat} visualises the simulation results, which show that for small standard deviations, the rainfall threshold $A_{\max}$ is close to its lower bound, before an increase in the kernel width causes it to peak before slowly decreasing as the kernel widths are further increased. For very narrow dispersal kernels very little spatial interaction takes place. In particular, as $\sigma \rightarrow 0$, the kernel functions tend to the delta function $\delta(x)$ centred at $0$ and therefore the integrodifference system \eqref{eq: Klausmeier integrodifference} becomes
\begin{align*}
u_{n+1}(x) &= u_n(x)+C\left(u_n(x)^2w_n(x)-Bu_n(x)\right), \\
w_{n+1}(x) &= w_n(x)+D\left(A-u_n(x)^2w_n(x)-w_n(x)\right).
\end{align*} 
For this system, the steady state $(\overline{u_2},\overline{w_2})$, which was randomly perturbed to set the initial condition of the system in the simulation, is always stable. Therefore, no patterns exist and $A_{\max} = 2B$ is the minimum value of the rainfall parameter for which vegetation is growing uniformly, recalling that for $A<2B$, the steady state $(\overline{u_2},\overline{w_2})$ does not exist. Further, away from $\sigma = 0$, a change in kernel width only has very little effect on $A_{\max}$, an indication that an increase to the width of the plant dispersal kernel has the opposite effect on the tendency to form patterns as an increase to the width of the water dispersal kernel.

\begin{figure}
	\centering
	\includegraphics[width=0.5\textwidth]{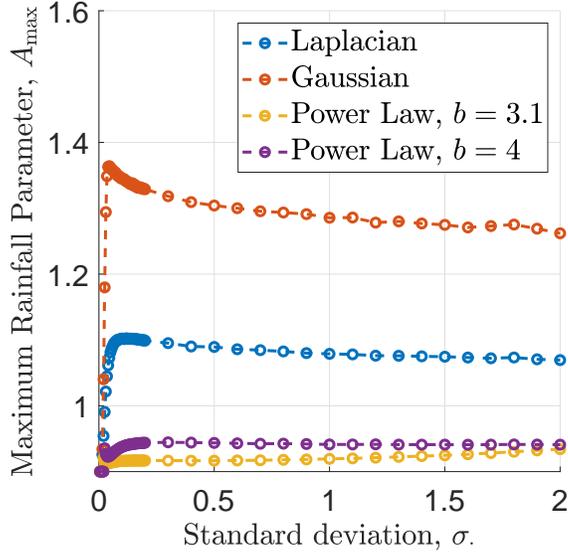}
	\caption{The maximum rainfall parameter $A_{\max}$ under simultaneous changes of the dispersal kernels. This figure visualises variations of $A_{\max}$ against simultaneous variations of both kernel functions. The standard deviation on the abscissa refers to the plant dispersal kernel $\phi$, the width of the water dispersal kernel $\phi_1$ is set to $a_2=0.1a$. The rainfall threshold is determined up to an interval of length $10^{-4}$ for $\sigma_\phi = \{0.01, 0.02, \dots, 0.05, 0.1, 0.2, \dots, 2\}$. The parameter values used for this simulation are $B=0.45$, $\ell=\ell_1=0.5$, $m=5$ }\label{fig:difference equations: simulations: Amax all kernels flat}
\end{figure}

To test this hypothesis, we investigate changes in the system's behaviour as individual kernel functions are changed. First, we consider how the critical rainfall parameter $A_{\max}$ is affected by a change of the shape of the dispersal kernel $\phi$ in the plant equation \eqref{difference equations: simulations: discretised system plants}. The result (see Figure \ref{fig:difference equations: simulations: Amax plant kernels flat}) is consistent with results of the integro-PDE model \cite{Eigentler2018nonlocalKlausmeier} on sloped ground. Firstly, an increase in the width of the plant dispersal kernels reduces the size of the parameter region supporting pattern onset, where changes for larger values of the standard deviation $\sigma_\phi$ are much smaller than close to $\sigma_\phi=0$. Identical to the nonlocal Klausmeier model, a trend that for small standard deviations those kernel functions that decay algebraically at infinity predict a lower value of $A_{\max}$ than those decaying exponentially, and vice versa for larger kernel widths, is also observed in these simulations.

\begin{figure}
	\centering
	\subfloat[Changes to plant dispersal kernel only\label{fig:difference equations: simulations: Amax plant kernels flat}]{\includegraphics[width=0.48\textwidth]{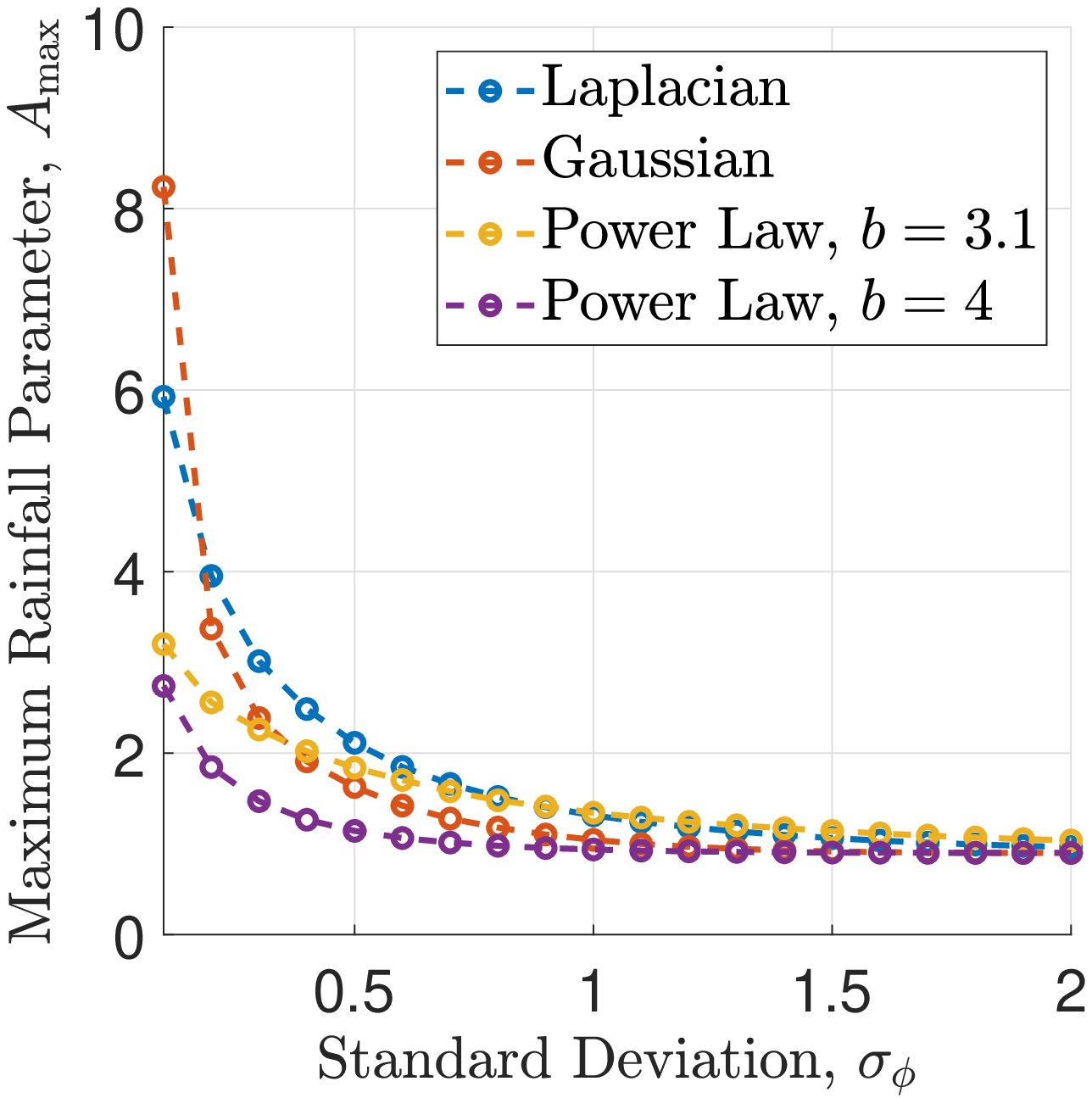}}
	\subfloat[Changes to water dispersal kernel only\label{fig:difference equations: simulations: Amax water kernels symmetric flat}]{\includegraphics[width=0.48\textwidth]{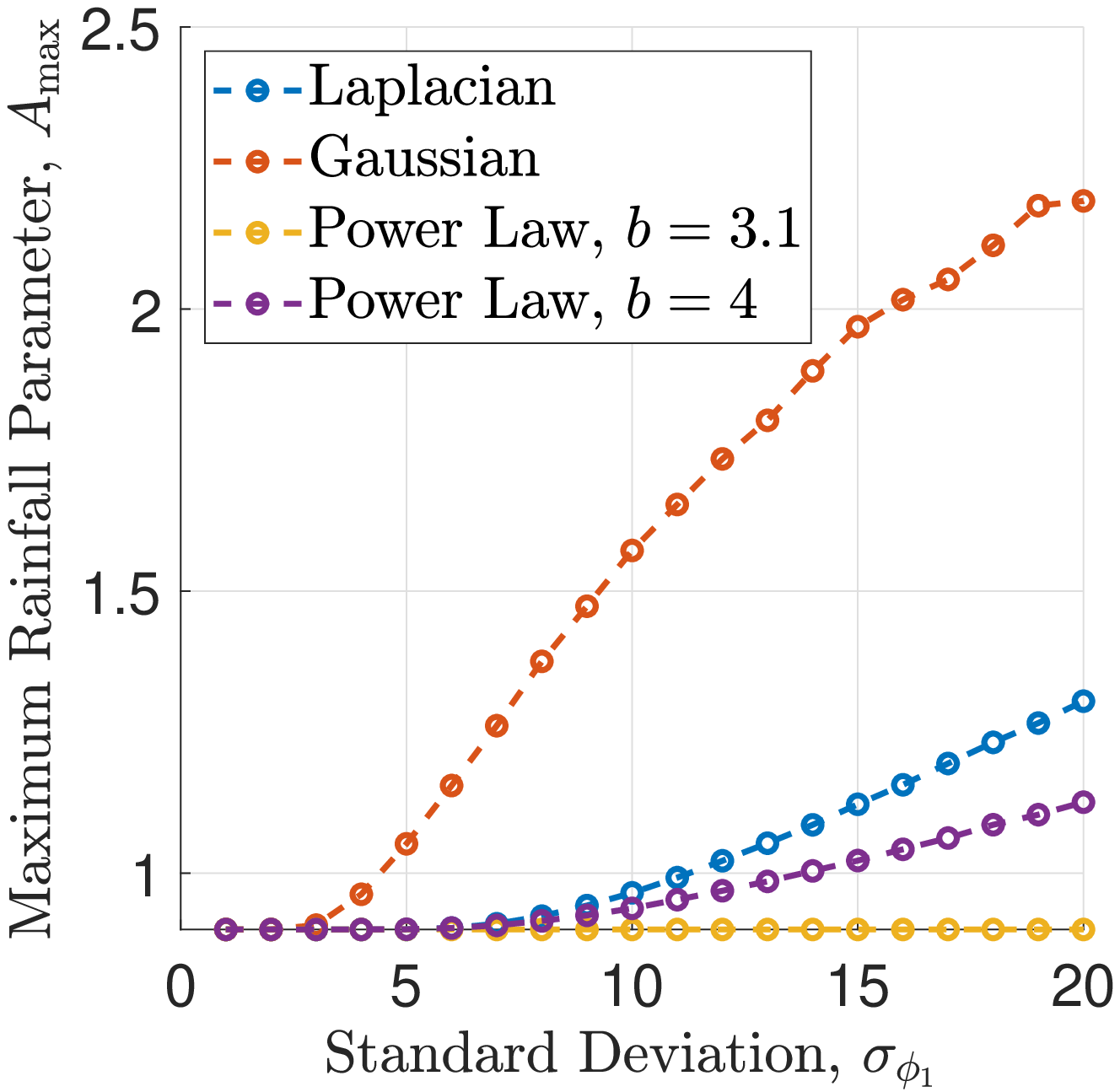}}
	\caption{The maximum rainfall parameter $A_{\max}$ under separate variations of the dispersal kernels.. Part (a) shows $A_{\max}$ up to an interval of length $10^{-4}$ with varying width ($\sigma_\phi = \{0.05, 0.1, 0.2, \dots, 2 \}$) and shape of the plant dispersal kernel $\phi$, while (b) visualises the effects of changes in the water dispersal kernel $\phi_1$. The latter was simulated for a larger range of the kernel's standard deviation $\sigma_{\phi_1}$, specifically $\sigma_{\phi_1} = \{1, 2, \dots, 20 \}$, to account for the choice of $a_2=0.1a$ in the previous simulation. Also in (b) $A_{\max}$ is determined up to an interval of length $10^{-4}$. The widths of the fixed kernels are set to $a_2 = 0.1$ (a) and $a=1$ (b), respectively. The other parameter values used both simulations are $B=0.45$, $\ell=\ell_1=0.5$, $m=5$.}
\end{figure}

Next, we perform a similar analysis for the symmetric water dispersal kernel $\phi_1$. To be consistent with the setting $a_2 = 0.1a$ in the simulation for the simultaneous change of the kernel functions, we consider a larger range of $\sigma_{\phi_1}$ for this simulation. The results (Figure \ref{fig:difference equations: simulations: Amax water kernels symmetric flat}) show that for narrow kernels, $A_{\max}$ is close to its minimum $A=2B$, i.e. the rainfall interval supporting pattern formation is very small. In particular, as $\sigma_{\phi_1}\rightarrow 0$, $A_{\max} \rightarrow 2B$ and no patterns can occur. For the Laplace kernel, this can also be shown using linear stability analysis. If $\sigma_{\phi_1} =  0$, then $\widehat{\phi_1} \equiv 1$ and thus the Jacobian \eqref{eq:difference: linear stability Jacobian} becomes 
\begin{align*}
J = \mattwo{C\widehat{\phi}(k)\alpha}{C\widehat{\phi}(k)\beta}{D\gamma}{D\delta}.
\end{align*}
Further, the stability condition is
\begin{align*}
k^2> \frac{BCa^2\left(A^2-A\sqrt{A^2-4B^2} -4B^2 \right)}{A^2-A\sqrt{A^2-4B^2}}.
\end{align*}
The right hand side is negative and thus the steady state is always stable to spatially heterogeneous perturbations. An increase of the kernel width then causes an increase in the rainfall threshold $A_{\max}$, where those kernels that decay exponentially at infinity, yield a larger increase than those decaying algebraically.

The results above confirm that the plant dispersal kernel $\phi$ and the water dispersal kernel $\phi_1$ have opposite effects on the rainfall threshold $A_{\max}$. While an increase in the width of the plant dispersal inhibits the onset of patterns, an increase in the standard deviation of the water dispersal kernel increases the tendency to form patterns. This explains the nearly constant value of $A_{\max}$ in the simulations in which both kernel functions are varied simultaneously. Consequently, these results suggest that it is the ratio of plant dispersal to water dispersal, i.e. the ratio $\sigma_{\phi}/\sigma_{\phi_1}$ that controls the tendency to form patterns. An increase in the ratio inhibits the onset of patterns, while a decrease has the opposite effect.

\section{Discussion}\label{sec: difference: Discussion}

The deliberately basic description of the plant-water dynamics in semi-arid environments by the Klausmeier model provides a rich framework for model extensions to address a range of different features of dryland ecosystems and their effects on vegetation patterns. Extensions include cross advection due to decreased surface water run-off resulting from an increase in infiltration in biomass patches \cite{Wang2019}; terrain curvature \cite{Gandhi2018}; nonlocal dispersal of seeds \cite{Eigentler2018nonlocalKlausmeier, Bennett2019}; secondary seed dispersal due to overland water flow \cite{Consolo2019a}; nonlocal grazing effects \cite{Siero2019, Siero2018}; explicit modelling of a population of grazers \cite{Fernandez-Oto2019a}; local competition between plants \cite{Wang2018a}; the inclusion of autotoxicity \cite{Marasco2014}; multispecies plant communities \cite{Eigentler2019Multispecies, Ursino2016, Callegaro2018} and seasonality and intermittency in precipitation \cite{Ursino2006, Eigentler2018impulsiveflat}. One aspect that has not yet been considered in this context is the seasonal separation of plant growth and seed dispersal. In this paper we have considered the synchronised and seasonal occurrence of nonlocal seed dispersal through a system of integrodifference equations based on the Klausmeier reaction-advection-diffusion system.

While an integrodifference system cannot explicitly quantify the temporal separation of seed dispersal occurrences,  the model's derivation and an associated convergence result (Proposition \ref{prop: Difference: limit case}) yield a parameter setting in which the length of the growth phase between dispersal stages can be accounted for. However, the main result of the linear stability analysis of the integrodifference model in this paper (Proposition \ref{prop: Difference: LinStab: diffusion threshold flat ground}) shows that conditions for pattern onset in the integrodifference model \eqref{eq: Klausmeier integrodifference} are independent of the temporal separation of seed dispersal from plant growth. Moreover, due to the model's derivation form the Klausmeier model \eqref{eq: Intro Klausmeier local}, the pattern onset conditions for both models are equivalent.

Some semi-arid environments in which vegetation patterning is a common phenomenon are characterised by large temporal and in particular seasonal fluctuations in their environmental conditions \cite{Noy-Meir1973, Chesson2004}. For example, observed patterns in Spain, Israel and North America are all located in Mediterranean climate zones \cite{Peel2007}, in which precipitation mainly occurs during winter, while during the summer months little or no rainfall occurs. {\color{changes}By} contrast, most mathematical models describing these ecosystems employ partial differential equations. While PDE models provide a rich framework for mathematical model analysis, their use is based on the simplifying assumption that all processes occur continuously in time. The results presented in this paper emphasise the importance and significance of results obtained from such models. In the context of seed dispersal, the biologically more realistic temporal separation of plant growth and seed dispersal has no effect on the conditions for pattern onset to occur. We thus conclude that the results obtained for the Klausmeier PDE model are robust to changes in the temporal properties of seed dispersal processes and that the assumption of continuous seed dispersal provides a sufficiently accurate description.

The parameter setting used to establish a connection between the Klausmeier model \eqref{eq: Intro Klausmeier local} and the integrodifference model \eqref{eq: Klausmeier integrodifference} couples the scale parameter $a$ of the seed dispersal kernel to other model parameters. If, however, a more general parameter setting is considered, then the effects of changes to the average seed dispersal distance and the shape of the seed dispersal kernel can be analysed numerically. Our results, which are in full agreement with an earlier investigation of the nonlocal Klausmeier model \eqref{eq: Intro Klausmeier nonlocal} \cite{Eigentler2018nonlocalKlausmeier}, show that seed dispersal over longer distances inhibits the formation of patterns (Figure \ref{fig:difference equations: simulations: Amax plant kernels flat}). Indeed, the threshold $A_{\max}$ on the rainfall parameter above which no pattern onset occurs, tends to $A_{\min}$, the minimum rainfall level required for the existence of a nontrival spatially uniform equilibrium, as dispersal distances become sufficiently large. Nevertheless, many plant species in semi-arid ecosystems have developed antitelechoric mechanisms which inhibit long range seed dispersal \cite{Ellner1981, RheedevanOudtshoorn2013}. While in the context of this paper this may appear as an evolutionary disadvantage, the development of narrow seed dispersal kernels is a side effect of other adaptations such as the development of seed containers as a protection to predation \cite{Ellner1981}. This suggests the existence of an evolutionary trade-off between seed dispersal distance and plant mortality. A numerical study of the threshold $A_{\max}$ in the $\sigma_\phi$-$B$ parameter plane (Figure \ref{fig: Difference: Discussion: Amax in sigma B plane}) gives some useful insight into this. The trade-off would restrict parameters to some increasing curve in the $\sigma_\phi$-$B$ parameter plane. Depending on the exact functional form of such a trade-off, a decrease in the seed dispersal distance $\sigma_\phi$ may cause a reduction in the precipitation threshold $A_{\max}$, if the trade-off implies a sufficiently large simultaneous decrease in the plant mortality rate $B$. A lower $A_{\max}$ value corresponds to an inhibition of pattern onset. We thus conclude that our model can capture the evolutionary advantage associated with the development of protective antitelechoric mechanisms if the trade-off between seed dispersal distance $\sigma_\phi$ and plant mortality $B$ is chosen appropriately, but emphasise that we are not aware of any data that provides quantitative information on the exact form of this trade-off.


\begin{figure}
	\centering
	\includegraphics[width=0.6\textwidth]{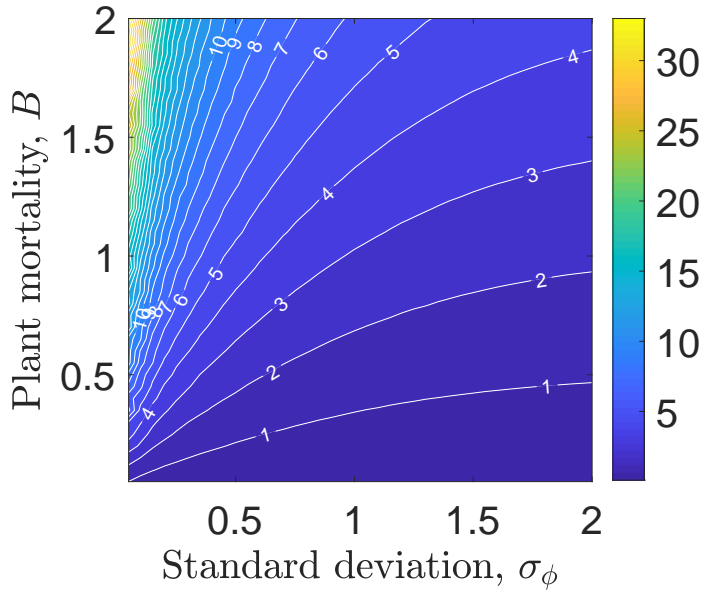}
	\caption{The threshold $A_{\max}$ in the $\sigma_\phi$-$B$ parameter plane. The numerically obtained rainfall threshold $A_{\max}$ is shown in the $\sigma_\phi$-$B$ parameter plane as a contour plot, where $\sigma_\phi$ denotes the standard deviation of the plant dispersal kernel $\phi$. It was obtained on the spatial grid $\{0.05,0.1,\dots,1.95,2\}\times\{0.05,0.1,\dots,1.95,2\}$ for the Laplace kernel \eqref{eq:difference equations: simulations: Laplacian} and $a_2 = 0.1$, $\ell = \ell_1 = 0.5$, $m=5$. We speculate that there may be an evolutionary trade-off between dispersal distance and resistance to predation, which would restrict parameters to an increasing curve in the $\sigma_\phi$-$B$ plane.}\label{fig: Difference: Discussion: Amax in sigma B plane}
\end{figure}

Our results further indicate that the shape of the seed dispersal kernel, and in particular its decay at infinity, has a significant effect on the onset of patterns. Fat-tailed kernels, for example, that account for a higher proportion of long-range dispersal events, yield a lower level of $A_{\max}$ than kernel functions with exponential decay at infinity for a sufficiently small fixed standard deviation. This highlights the importance of obtaining knowledge of seed dispersal behaviour of plant species, a property that depends on both species and the environment (e.g. seed dispersal agent) \cite{Bullock2017}. 

In our integrodifference model \eqref{eq: Klausmeier integrodifference}, we model the redistribution of water through a convolution similar to the modelling of the seed dispersal process. This nonlocal description can account for overland water flow from bare ground to biomass patches across larger distances during precipitation events. It does, however, rely on the assumption that the soil's properties enhance overland water flow in regions of low biomass. Some ecosystems in semi-arid environments are characterised by soil conditions and soil types (e.g. sand) for which this assumption is invalid \cite{Valentin1999}. The formation of vegetation patterns under such environmental conditions can, however, be explained by other mechanisms, such as laterally extended root networks \cite{Meron2018}. The integrodifference model presented in this paper is based on the assumption that little or no water infiltration occurs in regions of low biomass, and that the overland water flow towards regions of high biomass induced by this soil property is the main mechanism causing the self-organisation into patterns. In this context, our results show that water redistribution over longer distances yields the onset of patterns at higher precipitation levels (Figure \ref{fig:difference equations: simulations: Amax water kernels symmetric flat}). This is due to the enhancement of the pattern-inducing vegetation-infiltration feedback. Existing biomass patches deplete the water density locally, while regions of bare soil retain a higher water levels. Hence, any redistribution of water has a homogenising effect on the water density which yields to a redistribution of the limiting resource from areas of low biomass to areas of high biomass. An increase in the spatial range of the water redistribution kernel thus strengthens the pattern-inducing feedback and causes pattern onset under larger precipitation volumes.

The work in this paper shows that the description of seed dispersal as a synchronised event during a phase in which no plant growth occurs does not affect the condition for pattern onset compared to the continuous description of seed dispersal in the Klausmeier model \eqref{eq: Intro Klausmeier local}. The stability of spatial patterns is equally important. A natural area of future work would therefore be an analysis of pattern stability in the integrodifference model \eqref{eq: Klausmeier integrodifference} comparing results with stability results for the local Klausmeier model \cite{Sherratt2007} and the nonlocal Klausmeier model \cite{Bennett2019}. For PDE models, the stability of spatial patterns can be determined through a calculation of their spectra. For this, a method based on numerical continuation has been developed by Rademacher et al. \cite{Rademacher2007} (for details see \cite{Rademacher2007, Sherratt2012}). For integrodifference equations, however, we are not aware of any methods that allow the determination of the stability of a patterned solution. 

The integrodifference model \eqref{eq: Klausmeier integrodifference} not only splits the dynamics of the plant population into separate growth and dispersal stages, but also that of the water dynamics into a water consumption stage and a water redistribution stage. In the model, spatial redistribution of water is synchronised with seed dispersal. This can provide an adequate description for species such as \textit{Mesembryanthemum crystallinum} and \textit{Mesembryanthemum nodiflorum}, which synchronise their seed dispersal with the beginning of the rain season \cite{Navarro2009}, but cannot provide a description of seed dispersal during drought periods or of water flow at any other time during the rain season. While a description of the water flow dynamics during precipitation events in the context of a vegetation model has been proposed by Siteur et al. \cite{Siteur2014a}, the exact dynamics on flat ground are the subject of ongoing research (e.g. \cite{Rossi2017, Thompson2011, Wang2015}) and could be utilised in a future extension of the integrodifference model \eqref{eq: Klausmeier integrodifference}. The description of the water density as one single variable would, however, be prohibitive for such an approach. Instead a distinction between surface water and soil moisture, such as in the Rietkerk et al. model \cite{HilleRisLambers2001, Rietkerk2002} or the Gilad et al. model \cite{Gilad2004}, needs to be made to distinguish between surface water flow processes and water uptake processes that take place in the soil.

{\color{changes}The integrodifference model \eqref{eq: Klausmeier integrodifference} and its analysis presented in this paper is restricted to a one-dimensional space domain, motivated by the original formulation of the Klaumeier model and its mathematical accessibility \cite{Klausmeier1999}. However, the consideration of a second space dimension is expected to give more insights into the ecohydroglogical dynamics, in particular on pattern existence and stability. For example, in related PDE models on two-dimensional space domains, different types of patterned solutions exist (gap patterns, labyrinth patterns, striped patterns and gap patterns) and phase transitions along the precipitation gradient can be investigated \cite{Meron2012}. Moreover, even on sloped terrain, the impact of the consideration of a two-dimensional domain is significant, as the analysis on a one-dimensional domain may overestimate the size of the patterns' stability regions \cite{Siero2015}. The analysis of the integrodifference model \eqref{eq: Klausmeier integrodifference} on a two-dimensional domain presents a considerable challenge, in particular if one would want to obtain a wavenumber-independent results analogous to Prop. \ref{prop: Difference: LinStab: diffusion threshold flat ground}. Nevertheless, this would be a natural area of potential future work to further disentangle the complex ecosystem dynamics.}

Finally, we remark that the integrodifference model \eqref{eq: Klausmeier integrodifference} describes the discrete structure of plant growth mechanisms caused by the seasonality in precipitation. However, it does not capture the dynamics specific to drought periods between rainfall events and is thus only able to provide an insight into effects of accumulated rainfall volume rather than the temporal separation of precipitation seasons. In separate work, we account for a combination of rainfall, plant growth and seed dispersal pulses with the continuous nature of plant loss and water evaporation and drainage, using an impulsive model \cite{Eigentler2018impulsiveflat}. Such models combine partial differential equations with integrodifference equations (see for example \cite{Wang2018} for an impulsive model in the context of predator-prey dynamics with synchronised predator reproduction). The impulsive model has its own limitations as it can only take into account a periodic separation of precipitation events, but not any seasonal patterns. A potential area of future work therefore consists of a combination of these approaches to describe both the seasonal and intermittent nature of rainfall in semi-arid climate zones.

\section*{Acknowledgements}
Lukas Eigentler was supported by The Maxwell Institute Graduate School in Analysis and its Applications, a Centre for Doctoral Training funded by the UK Engineering and Physical Sciences Research Council (grant EP/L016508/01), the Scottish Funding Council, Heriot-Watt University and the University of Edinburgh.

\bibliography{bibliography.bib}

\end{document}